%% file: outerplanar-arxiv2.tex
\documentclass[11pt,a4paper]{article}

\usepackage{amsmath,amsthm,amsfonts,amssymb,graphics,color}

\usepackage{boxedminipage}
\newtheorem{theorem}{Theorem}
\newtheorem{lemma}{Lemma}

\newcommand{\NP}{{\sf NP}}

\newcommand{\dist}{{\rm dist}}

\newcommand{\tw}{{\mathbf{tw}}}
\newcommand{\pw}{{\mathbf{pw}}}
\newcommand{\ad}{{\rm ad}}

\newcommand{\problemdef}[3]{
	\begin{center}
		\begin{boxedminipage}{.99\textwidth}
			\textsc{{#1}}\\[2pt]
			\begin{tabular}{ r p{0.8\textwidth}}
				\textit{~~~~Instance:} & {#2}\\
				\textit{Question:} & {#3}
			\end{tabular}
		\end{boxedminipage}
	\end{center}
}

\title{Algorithms for Outerplanar Graph Roots \\ and Graph Roots of Pathwidth at Most~$2$\thanks{
This paper received support from the Research Council of Norway via the project ``CLASSIS'' and the Leverhulme Trust via Grant RPG-2016-258. An extended abstract of it appeared in the proceedings of WG 2017~\cite{GolovachHKLP17}. }}

\author{
Petr A. Golovach\thanks{Department of Informatics, University of Bergen, PB 7803
N-5020 Bergen, Norway, \texttt{\{petr.golovach,pinar.heggernes,paloma.lima\}@uib.no}}
\addtocounter{footnote}{-1}
\and 
Pinar Heggernes\footnotemark{}
\and
Dieter Kratsch\thanks{Universit\'e de Lorraine, LITA, Metz, France, \texttt{dieter.kratsch@univ-lorraine.fr}}
\addtocounter{footnote}{-2}
\and
Paloma T. Lima\footnotemark{}
\addtocounter{footnote}{1}
\and
Dani{\"e}l Paulusma\thanks{Department of Computer Science, Durham University, Durham DH1 3LE, UK,
\texttt{daniel.paulusma@durham.ac.uk}}
}

\date{}

\begin{document}
\pagestyle{plain}

\maketitle

\begin{abstract}
Deciding if a graph has a square root is a classical problem, which has been studied extensively both from graph-theoretic and algorithmic perspective. As the problem is \NP-complete, substantial effort has been dedicated to determining the complexity of deciding if a graph has a square root belonging to some specific graph class~${\cal H}$. There are both polynomial-time solvable and \NP-complete results in this direction, depending on ${\cal H}$. 
We present a general framework for the problem if ${\cal H}$ is a class of sparse graphs.
This enables us to generalize a number of known results and to
give polynomial-time algorithms for the cases where ${\cal H}$ is the class of outerplanar graphs and
${\cal H}$ is the class of graphs of pathwidth at most~$2$. 
\end{abstract} 

\section{Introduction}\label{s-intro}

Squares and square roots of graphs form a classical and well-studied topic in graph theory, which has also attracted significant attention from the algorithms community. A graph $G$ is the {\it square} of a graph $H$ if $G$ and $H$ have the same vertex set, 
and two vertices are adjacent in $G$ if and only if the
distance between them is at most~$2$ in $H$. This situation is denoted by $G = H^2$, and $H$ is called a {\it square root} of $G$. A square root of a graph need not be unique; it might even not exist. That is, there are graphs without square roots, graphs with a unique square root, and graphs with several different square roots. Characterizing and recognizing graphs with square roots has therefore been an intriguing and important graph-theoretic problem for more than 50 years (see e.g.~\cite{Ge68,Mukhopadhyay67,RossH60}).

In 1967, Mukhopadhyay~\cite{Mukhopadhyay67} proved that a graph $G$ on vertex set $\{v_1,\ldots,v_n\}$ has a square root if and only if $G$ contains complete subgraphs $\{K^1,\dots,K^n\}$, such that each $K^i$ contains $v_i$, and vertex $v_j$ belongs to $K^i$ if and only if $v_i$ belongs to $K^j$. Unfortunately, this characterization does not yield a polynomial-time algorithm for deciding whether~$G$ has a square root. 
This problem is called the {\sc Square Root} problem.
In 1994, Motwani and Sudan~\cite{MotwaniS94} proved that {\sc Square Root} is \NP-complete.

Motivated by its computational hardness, special cases of {\sc Square Root} have been studied where the input graph $G$ belongs to a particular graph class. 
It is known that  {\sc Square Root} is  polynomial-time 
solvable on planar graphs~\cite{LinS95}, and more generally, on every non-trivial minor-closed graph class~\cite{NT14}. Polynomial-time algorithms also exist if the input graph $G$ belongs to one of the following graph classes:
block graphs~\cite{LeT10}, line graphs~\cite{MOS14}, trivially perfect graphs~\cite{MilanicS13}, threshold graphs~\cite{MilanicS13},
graphs of maximum degree~6~\cite{CochefertCGKP13}, 
graphs of maximum average degree smaller 
than $\frac{46}{11}$~\cite{GKPS16b}\footnote{The average degree of a graph $G$ is defined as $\ad(G)=\frac{1}{|V_G|}\sum_{v\in V_G}d_G(v)=\frac{2|E_G|}{|V_G|}$. The maximum average degree of $G$ is then defined as $\max\{\ad(H)\; |\; H\text{ is a subgraph of }G\}$.}
graphs with clique number at most~3~\cite{GKPS16b}, and 
graphs with bounded clique number and no long induced path~\cite{GKPS16b}.
On the negative side, {\sc Square Root} is \NP-complete on chordal graphs~\cite{LauC04}. 
There also exist a number of parameterized complexity results for the problem~\cite{CCGKP,GKPS16}.

The intractability of {\sc Square Root} has also been attacked by restricting properties of the square root.
In this case, the input graph $G$ is an arbitrary graph, and the question is whether $G$ has a square root that belongs to some graph class~${\cal H}$ specified in advance. 
This problem is called  
{\sc $\mathcal{H}$-Square Root}, and this is the problem which we focus on in this paper. 

Significant advances have also been made on the complexity of {\sc $\mathcal{H}$-Square Root}.
Previous results show that {\sc ${\cal H}$-Square Root} 
is polynomial-time solvable for the following graph classes ${\cal H}$: 
trees~\cite{LinS95}, 
proper interval graphs~\cite{LauC04}, 
bipartite graphs~\cite{Lau06}, 
block graphs~\cite{LeT10}, 
strongly chordal split graphs~\cite{LeT11}, 
ptolemaic graphs~\cite{LOS15},  
3-sun-free split graphs~\cite{LOS15}, 
cactus graphs~\cite{GKPS16b}, cactus block graphs~\cite{Du17} and 
graphs with girth at least~$g$ for any fixed $g\geq 6$~\cite{FarzadLLT12}. The result for 3-sun-free split graphs was extended to a number of other subclasses of split graphs in~\cite{LOS}.  We observe that 
 if  $\mathcal{H}$-{\sc Square Root} is polynomial-time solvable for some class~${\cal H}$, then this does not automatically imply 
 that  $\mathcal{H'}$-{\sc Square Root} is polynomial-time solvable for a subclass
$\mathcal{H}'$ of~$\mathcal{H}$. 

On the negative side, {\sc ${\cal H}$-Square Root} remains \NP-complete for each of the following classes ${\cal H}$: 
graphs of girth at least~5~\cite{FarzadK12},
graphs of girth at least~4~\cite{FarzadLLT12},
split graphs~\cite{LauC04}, and 
chordal graphs~\cite{LauC04}.
All known \NP-hardness constructions involve dense graphs \cite{FarzadK12,FarzadLLT12,LauC04,MotwaniS94}, and the square roots that occur in these constructions are dense as well. This, in combination with the aforementioned polynomial-time results, leads to  our underlying research question:

\medskip
\noindent
{\it Is  {\sc ${\cal H}$-Square Root} polynomial-time solvable for 
every
sparse graph class $\mathcal{H}$?}

\medskip
\noindent
\subsection*{Our Results}
We give further evidence for the above question by proving that $\mathcal{H}$-{\sc Square Root} is polynomial-time solvable for two classes~${\cal H}$, namely when $\mathcal{H}$ is the class of outerplanar graphs, and when $\mathcal{H}$ is the class of graphs of pathwidth at most~$2$. 
Both classes are well studied. In particular, Syslo~\cite{Syslo79} 
characterized outerplanar graphs by a list of two forbidden minors, and Kinnersley and Langston~\cite{KinnersleyL94} gave a characterization of graphs of pathwidth at most~2 by a list of 110 forbidden minors (see~\cite{BH01,BHLY12} for an alternative approach). Outerplanar graphs have treewidth at most~2~\cite{Bo86}. However, they can have arbitrarily large pathwidth (as every tree is outerplanar and trees can have arbitrarily large pathwidth). Moreover, there exist graphs of pathwidth at most~2 that are not outerplanar; take, for instance, the complete bipartite graph $K_{2,t}$ on $t+2$ vertices for any $t\geq 3$.

The proofs of our results rely on structural properties that are specific for outerplanar graphs or graphs of pathwidth at most~2, respectively. However, despite the fact that the two classes are incomparable, the approach to obtain polynomial-time algorithms for each of them is
based  on the same general framework.
The basic idea is to use appropriate polynomial-time reduction rules, in which we try to recognize edges of the input graph~$G$ that belong to any square root or to no square root of $G$ at all.  
The goal is to obtain a graph whose treewidth is bounded by a constant, which enables us to solve the problem in polynomial time after expressing it in Monadic Second-Order Logic and applying a classical result of Courcelle~\cite{Courcelle92}. 
This idea has been used before (see, for instance,~\cite{CochefertCGKP13,CCGKP,GKPS16b,GKPS16}), but in this paper we {\it formalize} the idea into a general framework. We discuss this framework in detail in Section~\ref{sec:idea}.

Sections~\ref{sec:outerplanar-root} and~\ref{sec:pw} are dedicated to outerplanar graphs and graphs of pathwidth at most~2, respectively. In each of these two sections, we first prove the necessary structural properties of the graph class followed by a description of the algorithm, proof of correctness and running time analysis. Afterwards we prove that our general framework
enables us to solve $\mathcal{H}$-{\sc Square Root} in polynomial time for every subclass~${\cal H}$ of outerplanar graphs or graphs of pathwidth at most~2, respectively, that satisfies the following two conditions:
\begin{itemize}
\item [(i)] ${\cal H}$ is closed under taking a subgraph, and
\item [(ii)] ${\cal H}$ can be defined in Counting Monadic Second-Order Logic.
\end{itemize}
To give a few examples, our results imply the aforementioned results for the cases where
${\cal H}$ is the class of forests~\cite{LinS95} or cactus graphs (graphs in which every edge belongs to at most one cycle)~\cite{GKPS16b}, which both are subclasses of outerplanar graphs that satisfy conditions~(i) an~(ii).
To give another example, a connected graph has pathwidth~1 if and only if it is a caterpillar (a tree which can be modified in a path after removing all vertices of degree~1). 
The problem of deciding if a graph has a square root that is a caterpillar can be solved in polynomial time via a straightforward adaptation of the algorithm of~\cite{LinS95,RossH60} for trees. As the class of unions of caterpillars satisfy conditions~(i) and~(ii), this also follows from our results. Moreover, graphs of bandwidth at most~2, or equivalently, graph of proper pathwidth at most~2~\cite{KS96} have pathwidth at most~2 and satisfy conditions~(i) and~(ii). Hence we can also recognize squares of such graphs in polynomial time due to our results.

\section{Preliminaries}\label{sec:defs}
We consider only finite undirected graphs without loops and multiple edges. 
We refer to the textbook by Diestel~\cite{Diestel10} for any undefined graph terminology.
In the remainder we let $G$ be a graph.  

We denote the vertex set of $G$ by $V_G$ and the edge set by $E_G$. 
We use $n$ to denote the number of vertices of a graph (if this does not create confusion).
The subgraph of $G$
induced by a subset $U\subseteq V_G$ is denoted by $G[U]$. 
The graph $G-U$ is the graph obtained from $G$ after removing the vertices of $U$. If $U=\{u\}$, we also write $G-u$. 
Similarly, we denote the graph obtained from $G$ by deleting a set of edges $S$, or a single edge $e$, by $G-S$ and $G-e$, respectively.

The \emph{distance} $\dist_G(u,v)$ between a pair of vertices $u,v\in V_G$ is the number of edges of a shortest path between them in~$G$. We write $\dist_G(v,U)=\min\{\dist_G(u,v)\mid u\in U\}$ for a set of vertices $U\subseteq V_G$. 
For a positive integer $r$ and $u\in V_G$, we write $N_G^r(u)=\{v\in V_G\mid \dist_G(u,v)=r\}$. For $r=1$, we write $N_G(u)$ instead of $N_G^1(u)$ and say that   
$N_G(u)$ is the \emph{open neighborhood} of $u$.
The \emph{closed neighbourhood} of a vertex $u\in V_G$ is defined as $N_G[u] = N_G(u) \cup \{u\}$. 
For $S\subseteq V_G$, we let $N_G(S)=(\bigcup_{v\in S}N_G(v))\setminus S$.
Two distinct vertices $u,v$ are said to be \emph{true twins} if $N_G[u]=N_G[v]$ and $u,v$ are \emph{false twins} if $N_G(u)=N_G(v)$.
A vertex $v$ is \emph{simplicial} if $N_G[v]$ is a {\it clique}, that is, if there is an edge between any two vertices of $N_G[v]$.
The \emph{degree} of a vertex
$u\in V_G$ is defined as $d_G(u)=|N_G(u)|$.
The maximum degree of $G$ is $\Delta(G)=\max\{d_G(v)\; |\; v\in V_G\}$.
A vertex of degree~1 is said to be a \emph{pendant} vertex of~$G$. 

Let $K_r$ denote the complete graph on $r$ vertices and $K_{r,s}$ the complete bipartite graph with partition classes of size $r$ and $s$, respectively.

A {\it (connected) component} of $G$ is a maximal connected subgraph.
A vertex $u$ is a \emph{cut vertex} of a graph $G$ 
if $G-u$ has more connected components than $G$.  A connected graph without cut vertices is said to be \emph{biconnected}.
An inclusion-maximal 
induced biconnected 
subgraph of $G$ is called a \emph{block} of $G$.  

The \emph{contraction} of an edge $uv$ of a graph $G$ is the operation that deletes the vertices $u$ and $v$ and replaces them by a vertex $w$ adjacent to every vertex of $(N_G(u)\cup N_G(v))\setminus\{u,v\}$.
A graph $G'$ is a contraction of a graph $G$ if $G'$ can be obtained from $G$ by edge contractions. A graph $G'$ is a \emph{minor} of $G$ if $G'$ can be obtained from $G$ by vertex deletions, edge deletions and edge contractions.

The syntax of \emph{Monadic Second-Order Logic} (MSO) on graphs includes 
\begin{itemize}
\item logical connectivities $\vee$, $\wedge$ and $\neg$,
\item variables for vertices, edges, sets of vertices and sets of edges,  
\item the quantifiers $\exists$ and $\forall$ that apply to variables,
\item the predicates $=$, $\in$, $\mathbf{adj}$ and $\mathbf{inc}$ for equality, inclusion of an element in a set, adjacency of vertices and incidence of a vertex with an edge, respectively.
\end{itemize}
\emph{Counting Monadic Second-Order Logic} (CMSO) is the extension of MSO with the predicate $\mathbf{card}_{q,p}(S)$ defined 
on sets for some integer constants~$p$ and~$q$ with $0\leq q<p$ and $p\geq2$, 
such that  $\mathbf{card}_{q,p}(S)={\sf true}$ if and only if $|S|\mod p=q$. 
For a CMSO formula $\varphi$ on graphs, we write $G\models \varphi$ to denote that $\varphi$ evaluates ${\sf true}$ on $G$.
We refer to the book of Courcelle and Engelfriet~\cite{CourcelleE12} for an introduction to MSO and CMSO.

We will use the following well-known fact  (see, for example,~\cite{CourcelleE12}).

\begin{lemma}\label{l-engel}
The property that a graph $G$ contains a fixed graph $F$ as a minor can be expressed in  MSO. 
\end{lemma}  

\subsection{Square Roots}\label{s-square}

For a positive integer $k$, the \emph{$k$-th power} of a graph $H$ is the graph $G=H^k$ with  
vertex set $V_G=V_H$,
such that every pair of distinct vertices $u$ and $v$ of $G$ are adjacent if and only if $\dist_H(u,v)\leq k$. 
If $k=2$, then $H^2$ is called a \emph{square} of~$H$, and $H$ is called a \emph{square root} of $G$ if $G=H^2$.

We say that a square root $H$ of a graph $G$ is \emph{minimal} if no proper subgraph of $H$ is a square root of $G$. 
We need two basic lemmas on minimal square roots. The first lemma follows immediately from the definition. We give a short proof for
the second lemma.

\begin{lemma}\label{obs:min}
Let $\mathcal{H}$ be a graph class closed under 
taking edge deletions and vertex deletions.
If  a graph~$G$ has a square root in $\mathcal{H}$, then $G$ has a minimal square root in $\mathcal{H}$.
\end{lemma}

\begin{lemma}\label{obs:triangle}
Let $H$ be a minimal square root of a graph $G$ that contains three vertices $u,v,w$ that are pairwise adjacent in $H$. Then $v$ or $w$ has a neighbour $x\neq u$ in $H$ such that $x$ is not adjacent to $u$ in $H$ 
and $x$ is adjacent to exactly one of $v,w$ in $H$.
\end{lemma}

\begin{proof}
As $H$ is a minimal square root of $G$, $H-vw$ is not a square root of $G$. 
Hence, there is an edge $xy\in E_G\setminus E_H$, such that $H$ has a unique $(x,y)$-path~$P$ of length~$2$ and $wv$ is an edge of this path. Therefore, exactly one of $v,w$ is adjacent to $x$ in~$H$.
As $uv$ and $uw$ are both edges in $E_H$, this means that $P=xvw$ or $P=xwv$ for some $x\neq u$. Note that $x$ is not adjacent to $u$, because, otherwise, either $P'=xuw$ or $P'=xuv$ would be the second $(x,y)$-path of length~2.
\end{proof}

We also need a lemma that is implicit in~\cite{GKPS16b}. This lemma enables us to identify some edges that are not included in any square root.

\begin{lemma}\label{lem:not-incl}
Let $x,y$ be two neighbours of a vertex $u$ in a graph $G$ that are of distance at least~$3$ in $G-u$. Then $ux,uy\notin E_H$ for any square root $H$ of $G$.
\end{lemma}

\begin{proof}
Suppose  $H$ is a square root of $G$. For contradiction, assume $ux\in E_H$. If $uy\in E_H$, then $xy\in E_G$ contradicting the assumption that $\dist_{G-u}(x,y)\geq 3$.  Hence $uy\notin E_H$. As $uy\in E_G$, there exists a vertex $z$ such that $uz,zy\in E_H$. If $z=x$, then $xy\in E_G$; a contradiction. If $z\neq x$, then $ux\in E_H$ and $uz\in E_H$ imply that
 $xz\in E_G$. Hence $xzy$ is a path in $G-u$ of length~2, and again we obtain a contradiction with our assumption that $\dist_{G-u}(x,y)\geq 3$. We conclude that $ux\notin E_H$ and for the same reason we obtain $uy\notin E_H$.
\end{proof}

\subsection{Treewidth and Pathwidth}\label{s-treepath}

A \emph{tree decomposition} of a graph $G$ is a pair $(T,X)$ where $T$
is a tree, 
whose vertices are called {\it nodes}, 
and $X=\{X_{i} \mid i\in V_T\}$ is a collection of subsets, called {\em bags},
of $V_G$ such that the following three conditions hold: 
\begin{itemize}
\item[i)] $\bigcup_{i \in V_T} X_{i} = V_G$;
\item[ii)] for all $xy \in E_G$, $x,y\in X_i$ for some  $i\in V_T$; and 
\item[iii)] for all $x\in V_G$, $\{ i\in V_T \mid x \in X_{i} \}$ induces a connected subtree of $T$.
\end{itemize}
The \emph{width} of a tree decomposition $(\{ X_{i} \mid i \in V_T \},T)$ is $\max_{i \in V_T}\,\{|X_{i}| - 1\}$. The \emph{treewidth} $\tw(G)$ of a graph $G$ is the minimum width over all tree decompositions of $G$.
If $T$ is a path, then we say that $(X,T)$ is a \emph{path decomposition} of $G$.
The \emph{pathwidth} $\pw(G)$ of $G$ is the minimum width over all path decompositions of $G$.
Notice that a path decomposition of $G$ can be seen as a sequence $(X_1,\ldots,X_r)$ of bags. 
We always assume that the bags $(X_1,\ldots,X_r)$ are distinct and  inclusion
incomparable, that is, there are no bags $X_i$ and $X_j$ such that $X_i\subset X_j$.

The next lemma gives two fundamental results on treewidth and pathwidth, which are due to Bodlaender, and Bodlaender and Kloks, respectively.

\begin{lemma}[\cite{Bodlaender96,BodlaenderK96}]\label{l-bod}
For every 
constant~$c$, it is possible to decide in linear time whether the treewidth or the pathwidth of a graph is at most~$c$.
\end{lemma}

We will also need the following two well-known lemmas. We refer to~\cite{Diestel10} for the first lemma.
Note that this lemma also holds if $H$ is a contraction of $G$ (as this immediately implies that $H$ is a minor of $G$).
We provide a proof of the second lemma. 
This lemma is also folklore, but might not have been stated in this way. In particular, 
we formulate it for arbitrary $k\geq 1$ instead of for $k=2$ only, as we will need this observation in general form to prove our results.

\begin{lemma}\label{obs:minor}
Let $G$ and $H$ be graphs.
If $H$ is a minor of $G$, then $\tw(H)\leq \tw(G)$ and $\pw(H)\leq \pw(G)$.
\end{lemma}

\begin{lemma}\label{lem:tw-power} 
For a graph $G$ and an integer $k\geq 1$, the following hold:
$\tw(G^k)\leq (\tw(G)+1)\Delta(G)^{\lfloor k/2\rfloor+1}$ and 
$\pw(G^k)\leq (\pw(G)+1)\Delta(G)^{\lfloor k/2\rfloor+1}.$
\end{lemma}

\begin{proof}
We show that $\tw(G^k)\leq (\tw(G)+1)\Delta(G)^{\lfloor k/2\rfloor+1}$. The proof of the second inequality uses the same arguments. The inequality is trivial if $k=1$ or $\Delta(G)\leq 1$. Assume that $k\geq 2$ and $\Delta(G)\geq 2$. Let also  $\ell=\lfloor k/2\rfloor$.

Let  $(T,X)$ be a tree decomposition of $G$ of minimum width.
For $i\in V_T$, we define 
$Y_i=\{y\in V_G\mid \dist_G(y,X_i)\leq \ell\}$ and $Y=\{Y_{i} \mid i\in V_T\}$.
We show that $(T,Y)$ is a tree decomposition of $G^k$ by proving that conditions (i)--(iii) of the definition of treewidth are satisfied.

\medskip
\noindent
(i).  As $X_i\subseteq Y_i \subseteq V_G$ for all $i\in V_T$, we obtain $V_G=\cup_{i\in V_T}X_i\subseteq \cup_{i\in V_T}Y_i\subseteq V_G$, so  $\cup_{i\in V_T}Y_i= V_G$.
As $V_{G^k}=V_G$, this means that $\cup_{i\in V_T}Y_i=V_{G^k}$, so
(i) holds.

\medskip
\noindent
(ii). Consider an edge $uv$ of $G^k$. By the definition of $G^k$, $G$ contains a $(u,v)$-path $P$ of length at most~$k$. 
Then $P$ has an edge $xy$ such that $\dist_G(u,x)\leq\ell$ and $\dist_G(y,v)\leq \ell$. Because $(T,X)$ is a tree decomposition of $G$, there is
a node~$i\in V_T$ such that $x,y\in X_i$. We find that $u,v\in Y_i$. Hence (ii) holds.

\medskip
\noindent
(iii). For contradiction, assume that  there is $v\in V_{G^k}$ such that the set $\{i\in V_T\mid v\in Y_i\}$ is disconnected. Then there exist two distinct nonadjacent nodes $i,j\in V_T$ such that $v\in Y_i$, 
$v\in Y_j$,
and $v\notin Y_h$ for every internal node $h$ of the unique $(i,j)$-path in $T$. Since $v\in Y_i$, there exists a vertex $x\in X_i$ such that $\dist_G(x,v)\leq\ell$. Similarly, there exists a 
vertex~$y\in X_j$
such that $\dist_G(y,v)\leq\ell$. Let $P_x$ and $P_y$ be shortest $(x,v)$-paths and $(y,v)$-paths in $G$ respectively. Then $G$ contains an $(x,y)$-path $P$ whose edges belong to $E_{P_x}\cup E_{P_y}$. 
Note that every vertex of $P$ is of distance at most~$\ell$ 
from $v$ in $G$.

Consider an arbitrary internal node $h$ of the unique $(i,j)$-path in $T$. Let $z\in X_h$. As $v\notin Y_h$, it follows that $\ell<\dist_G(v,X_h)\leq \dist_G(v,z)$. Hence, $z\notin V_P$. We conclude that $X_h\cap V_P=\emptyset$. This means that the bag $X_h$ does not separate $x$ and $y$ in $G$, 
which contradicts a basic property of a tree decomposition  (see, for example, Lemma~12.3.1~\cite{Diestel10}).

We now prove the bound on the width of $(T,Y)$. For $i\in V_T$, we find that
\[\begin{array}{lcl}
|Y_i| &\leq&|X_i|(1+\ldots+\Delta(G)^\ell)\\[5pt]
 &= &\displaystyle |X_i|\frac{\Delta(G)^{\ell+1}-1}{\Delta(G)-1}\\[13pt]
 &\leq &|X_i|\Delta(G)^{\ell+1}\\[5pt]
&\leq&(\tw(G)+1)\Delta(G)^{\ell+1}.
\end{array}\]
Hence
$\tw(G^k)\leq\max_{i\in V_T}|Y_i|-1\leq (\tw(G)+1)\Delta(G)^{\ell+1}$.
\end{proof}

We will also need the following characterization of graphs of pathwidth at most~$2$, which is due to Kinnersley and Langston (we do not specify the graphs on their list, as this is irrelevant for our purposes).

\begin{lemma}[\cite{KinnersleyL94}]\label{lem:pw2c}
A graph has pathwidth at most~$2$ if and only if does not contain a graph from a specific list of 110 graphs as a minor.
\end{lemma}

As mentioned we will also need the following classical result of Courcelle as a lemma.

\begin{lemma}[\cite{Courcelle92}]\label{l-courcelle}
For  every fixed integer~$k$ and every problem ${\cal P}$ expressible in CMSO, there  exists a
linear-time algorithm that solves ${\cal P}$ for the class of graphs of treewidth at most~$k$.
\end{lemma}

\subsection{Outerplanar Graphs}
A graph $G$ is \emph{planar} if $G$ admits a {\it planar} embedding, which is an embedding on the plane in such a way that the edges of $G$ only intersect at their end-points. 
A planar graph $G$ is 
{\it outerplanar} if it admits a planar embedding in which
all its vertices
belong to the outerface. 
When considering an outerplanar graph, we always assume that such an embedding is given.  

\begin{figure}[ht]
\centering
\scalebox{0.8}{
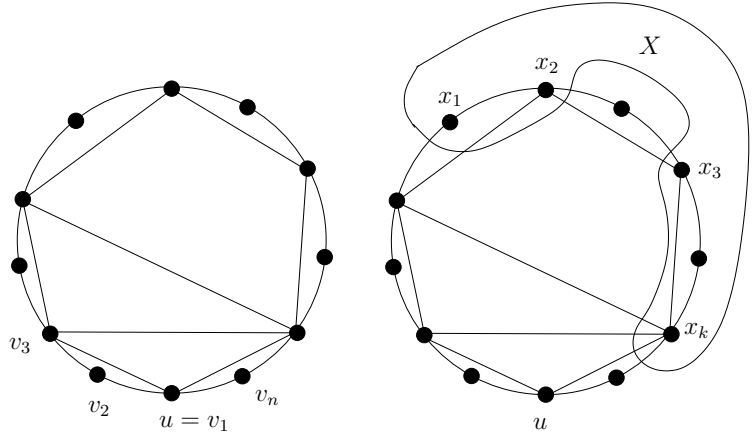}
\caption{A clockwise ordering of the vertices $v_1,\ldots,v_n$ of a biconnected outerplanar graph~$G$ with respect to vertex $u=v_1$ and a clockwise ordering of a set $X=\{x_1,\ldots,x_k\}$ of $G$ with respect to $u$. 
\label{fig:ord}}
\end{figure}

\begin{figure}[ht]
\centering
\scalebox{0.8}{
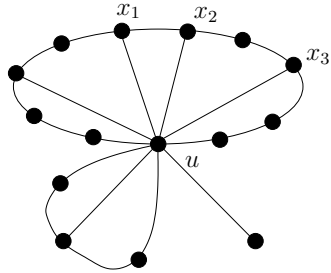}
\caption{An example of an outerplanar graph with a set $X=\{x_1,x_2,x_3\}$ that is consecutive with respect to $u$; note that 
$x_1$ and $x_2$ are consecutive with respect to $u$, just as $x_2$ and $x_3$, while $x_1$ and $x_3$ are not consecutive with respect to $u$.
\label{fig:ord2}}
\end{figure}

If $G$ is a planar biconnected graph different from $K_2$, then for any of its  
embeddings,
the boundary of each face is a cycle (see, e.g., \cite{Diestel10}). If $G$ is a biconnected outerplanar graph distinct from $K_2$, then the cycle $C$ forming the boundary of the external face is unique~\cite{Syslo79}. We call $C$ 
the \emph{boundary cycle} of $G$. 
Every vertex of $G$ belongs to $C$, and every edge of $G$ is either an edge of $C$ or a \emph{chord} of $C$, that is, its endpoints are vertices of $C$ that are non-adjacent in $C$. By definition, these chords are not intersecting in the embedding.   
We define the
\emph{clockwise ordering} of $C$ with respect to some vertex $u$ of $G$ as
the clockwise ordering 
of the vertices on $C$ starting from $u$. For a subset of vertices $X$, the 
\emph{clockwise ordering} of $X$ with respect to $u$ is the restriction of the clockwise ordering of $C$ to the vertices of $X$.
See Figure~\ref{fig:ord} for 
an example of these notions.

We use the above terms for blocks of an outerplanar graph that are distinct from~$K_2$.
 We say that  two distinct vertices $x,y\in N_G(u)$ are \emph{consecutive with respect to $u$} if $x$ and $y$ are in the same block $F$ of $G$ and  there are no vertices of $N_G(u)$ between $x$ and $y$  in the clockwise ordering  of the vertices of the boundary cycle of $F$ with respect to $u$.  For a set of vertices $X\subseteq N_G(u)$, we say that the 
vertices of 
$X$ are  \emph{consecutive with respect to $u$} if the vertices of $X$ are in the same block of $G$ and any two vertices of $X$ consecutive in the clockwise ordering of 
the vertices  of $X$ with respect to $u$ are consecutive with respect to $u$. 
See Figure~\ref{fig:ord2} for an illustration of these notions.
  
Sys{\l}o characterized the class of outerplanar graphs via a set of two forbidden minors. 
\begin{lemma}[\cite{Syslo79}]\label{lem:out-minor}
A graph is outerplanar if and only if it  does not contain $K_{2,3}$ or $K_4$ as a minor.
\end{lemma}
\noindent
We will also need the following well-known result.

\begin{lemma}[\cite{Bo86}]\label{obs:tw-op}
Every outerplanar graph has treewidth at most~$2$.
\end{lemma}

\section{General Algorithmic Approach} \label{sec:idea}

Our algorithms for deciding whether a graph has an outerplanar square root or a square root that has pathwidth at most~2, respectively, rely on similar ideas and concepts.

The
framework underlying these two algorithms is general and has the potential to be applicable to find other restricted square roots as well. This section is devoted to explain this framework. 

We observe that even though outerplanar graphs and graphs of pathwidth at most~2 have bounded treewidth, squares of such graphs {\it may} have arbitrarily large treewidth.
The basic idea is to reduce an input graph of our problem in polynomial time to a graph of bounded treewidth that is an instance of a closely related auxiliary problem.  After showing that this auxiliary problem can be expressed in MSO, we can then apply the well-known result of Courcelle~\cite{Courcelle92}.

Let ${\cal H}$ be the class of outerplanar graphs or graphs of pathwidth at most~2. Let $G$ be the input graph. In order to find out if $G$ has a square root  $H\in{\cal H}$, we modify $G$ using the following polynomial-time rules, 
which we apply exhaustively in the order below:

\begin{enumerate}
\item \textbf{Deleting irrelevant vertices.} 
We try to identify vertices that can be deleted from $G$, such that the resulting graph has a square root that belongs to ${\cal H}$ if and only if $G$ has a square root that belongs to ${\cal H}$. If ${\cal H}$ is the class of outerplanar graphs, then this step allows us to bound the number of 
true twins of a simplicial vertex of $G$.
If ${\cal H}$ is the class of graphs of pathwidth at most~2, then this step allows us to bound the number of true twins of a
vertex of~$G$. 

\item \textbf{Labeling edges.} 
We try to identify edges in $G$ that we can give a specific label. That is, we label an edge~$e$ of $G$ {\it red} if we can determine that $e$ belongs to every minimal
square root of $G$ and we label $e$ {\it blue} if we can determine that $e$ does not belong to any minimal square root of $G$.
We let $R$ and $B$ denote the sets of red and blue edges, respectively.

\item \textbf{Deleting irrelevant edges.} 
We determine a set $U\subseteq V_G$ such that for each $u\in U$, all edges incident to $u$ are labeled either red or blue.
We return a no-answer if there exist two red edges that form an induced  path of length~2 in $G$. Otherwise,
for each $u\in U$, we identify a set of edges in $G[N(u)]$  that we may remove from $G$. 
The crucial properties of every deleted edge $xy$ are that (i) $xy$ is not included in any minimal square root of $G$ and (ii) there are two red edges incident to $u$ 
in $G$ that form an $(x,y)$-path. 
\end{enumerate}

\noindent
After performing 
each of these three rules exhaustively in the given order,
we obtain a new graph $G'$. The crucial point here is that 
if the treewidth of $G'$ is greater than some constant~$c$ that does not depend on~$G$ or~$H$, then $G$ does not have a square root~$H$ such that $H\in {\cal H}$. 
Assume that the treewidth of $G'$ is at most~$c$.
Recall that in order to obtain $G'$ we 
deleted some edges of $G$. Therefore, a square root of $G$ is not necessarily a square root of $G'$ and, moreover, $G'$ may not even have a square root. Nevertheless, 
by using properties~(i) and~(ii) above, we can recover the structure of 
a square root of $G$.
More formally, we prove that $G$ has a square root in ${\cal H}$ if and only if $G'$ contains a subset $L\subseteq E_{G'}$ with the following properties:

\begin{itemize}
\item[(i)] $R\subseteq L$ and $B\cap L=\emptyset$;
\item[(ii)]  for every $xy\in E_{G'}$, $xy\in L$ or there exists  a vertex~$z\in V_{G'}$ with $xz,zy\in L$;
\item[(iii)] for every two distinct edges $xz,yz\in L$, $xy\in E_{G'}$ or there is a vertex $u\in U$ with $xu,uy\in R$; and
\item[(iv)] the graph $H=(V_G,L)$ belongs to ${\cal H}$.
\end{itemize}
In fact, $H=(V_G,L)$ is a square root of the graph obtained from $G$ via exhaustive application of the first rule.
Since $G'$ has bounded treewidth and properties (i)--(iv) can be expressed in MSO, we can test the existence of the set $L$ in polynomial time by applying the aforementioned result of Courcelle~\cite{Courcelle92}.

\section{Outerplanar Roots}\label{sec:outerplanar-root}

We say that a square root $H$ of $G$ is an \emph{outerplanar root} if $H$ is outerplanar, and we define the following problem:

\problemdef{\sc Outerplanar Root}{a graph $G$.}{does $G$ have an outerplanar root?}

The main result of this section is the following theorem.

\begin{theorem}\label{thm:outerplanar}
\textsc{Outerplanar Root} can be solved in $O(n^4)$ time.
\end{theorem}
\noindent
We first show a number of structural results in Section~\ref{sec:tech}.
We then use these results in the design of our polynomial-time algorithm for {\sc Outerplanar Root} in Section~\ref{sec:alg}.

\subsection{Structural Lemmas}\label{sec:tech}
As the class of outerplanar graphs is closed under 
vertex an edge deletions,
we may restrict ourselves to  minimal outerplanar roots by Lemma~\ref{obs:min}.

We start with the following lemmas.

\begin{lemma}\label{lem:adj}
Let $H$ be a minimal square root of a graph $G$, and let $u\in V_G$. If $x\in N_H(u)$ is not a pendant vertex of $H$, then  there is a vertex $y\in N_H^2(u)$
 that is adjacent to $x$ in $G$.
\end{lemma}

\begin{proof}
Since $x$ is not a pendant vertex of $H$, $x$ has at least one neighbour in $H$ distinct from $u$. If there exists a vertex $y\in  N_H(x)$ such that $y\notin N_H[u]$, then the claim holds. Assume that for every $y\in N_H(x)$ distinct from $u$, it holds that $y\in N_H(u)$. Consider such a neighbour $y$. By Lemma~\ref{obs:triangle}, $x$ or 
$y$ has a neighbour 
$y'$ in $H$ such that $y'\notin N_H(u)$. By our assumption on $x$, we find that this vertex cannot be $x$ and thus must be $y$.
Since $xy,yy'\in E_H$, we find that $xy'\in E_G$. 
Hence $y'\in N_G(u)\setminus N_H(u)=N_H^2(u)$ is adjacent to $x$ in $G$, as desired.
\end{proof}

Let $H$ be a minimal square root of a graph $G$ and let $u\in V_G$. We define the following set: 
$$S(H,u)=\{N_G(x)\cap N_H(u)\; |\; x\in N_H^2(u)\}.$$ 
We use $S(H,u)$ to detect edges with both endpoints in $N_H(u)$  that are excluded from every minimal square root of $G$.

\begin{lemma}\label{lem:non-edges}
Let $H$ be a minimal square root of a graph $G$, and let $u\in V_G$.  If for two distinct vertices $x,y\in N_H(u)$ there is no set $X\in S(H,u)$ such that $x,y\in X$, then $xy\notin E_H$.
\end{lemma}

\begin{proof}
Suppose that for two distinct vertices $x,y\in N_H(u)$, $xy\in E_H$. By Lemma~\ref{obs:triangle}, there is a vertex $z$ such that $z$ is adjacent to $x$ or $y$ in $H$, but $z$ is not adjacent to $u$ in $H$. We find that $x$ and $y$ are adjacent to $z$ in $G$ and, therefore,  $x,y\in N_G(z)\cap N_H(u)$. In other words, $x,y\in X=N_G(z)\cap N_H(u)\in S(H,u)$.
\end{proof}

We need the following two lemmas about the structure of $S(H,u)$ for minimal outerplanar roots. 

\begin{lemma}\label{lem:cons}
Let $H$ be a minimal outerplanar root of a graph $G$, and let $u\in V_G$.  Then, for each $X\in S(H,u)$, $X$ is consecutive with respect to $u$.
\end{lemma}

\begin{proof}
Let $X\in S(H,u)$ and consider a vertex~$x\in N_H^2(u)$ such that $X=N_G(x)\cap N_H(u)$. 
As $ux\in E_G$ but $ux\notin E_H$, we find that $x$ must be adjacent to a vertex of $N_H(u)$. Hence $X\neq \emptyset$.

If $|X|=1$, then the claims holds by definition. Assume that $|X|\geq 2$. 

We first observe that the vertices of $X$ are in the same block~$F$ of $H$.
This can be seen as follows.
Suppose $y,z\in X\subseteq N_H(u)$ are two vertices that are not in the same block of $H$. Then 
any vertex adjacent to $y$ and $z$ in $G$ must belong to $N_H[u]$. Hence $x\in N_H(u)$, which is not possible.

Let~$C$ be the boundary cycle of~$F$.
Assume that $X=\{x_1,\ldots,x_k\}$ and that the vertices are numbered in the clockwise order with respect to  $u$. 
Suppose that $X$ is not consecutive with respect to $u$. Then, by definition, 
there exists a vertex~$y\in N_H(u)$ that lies on $C$ between two
vertices $x_{i-1}$ and $x_i$ for some $i\in\{2,\ldots,k\}$, such that $y$ does not belong to $X$. The latter implies that 
$y\notin N_G(x)$. 
Since $x_{i-1}x$ and $x_ix$ belong to $E_G$ by definition of $X$, we find that $H$ has an
$(x_{i-1},x)$-path~$P_1$ and an $(x_i,x)$-path~$P_2$, each of length at most~$2$. The paths $P_1$ and $P_2$ do not contain $u$, due to the facts that $xu\notin E_H$ and the length of $P_1$ and $P_2$ is at most~$2$.
Moreover, $P_1$ and $P_2$ do not contain $y$, as $yx\notin E_G$ and both paths have length at most~$2$.

Let $Q_1$ be the subpath of $C$ from $x_{i-1}$ to $y$ that does not contain $x_i$, and let $Q_2$ be the subpath of $C$ from
$y$ to $x_i$ that does not contain $x_{i-1}$. First suppose that $x$ does not belong to $Q_1$ or $Q_2$.
 We contract all edges on $P_1$ and $P_2$ and every edge on $Q_1$ and $Q_2$. This yields a $K_4$ with vertices $u$, $x_{i-1}$, $x_i$ and $y$, contradicting Lemma~\ref{lem:out-minor}.
 Now suppose that $x$ belongs to $Q_1$ or $Q_2$, say to $Q_1$. We contract the edge $ux_{i-1}$, every edge on $Q_1$, every edge of $Q_2$ and every edge of $P_2$.
 This yields a $K_4$ with vertices $u$, $x$, $x_i$ and $y$, contradicting Lemma~\ref{lem:out-minor} again.
 We conclude that $X$ must be consecutive with respect to $u$.
\end{proof}

\begin{lemma}\label{lem:size}
Let $H$ be a minimal outerplanar root of a graph $G$, and let $u\in V_G$. Then any $X\in S(H,u)$ has size at most~$4$.
\end{lemma}

\begin{proof}
For contradiction, assume that there exists a set $X\in S(H,u)$ of size at least~5. Let $X=\{x_1,\ldots,x_k\}$ for some $k\geq 5$.
By definition,  $X=N_G(x)\cap N_H(u)$ for some vertex 
$x\in N_H^2(u)$.
By Lemma~\ref{lem:cons}, $X$ is consecutive with respect to $u$. We assume that $x_1,\ldots,x_k$ is the clockwise order of the vertices of $X$ along the boundary cycle $C$ of the block of $H$ with respect to $u$.  
  
First suppose that $x$ belongs to $C$. If $x$ lies before $x_3$ in the clockwise  ordering of the vertices of $C$ with respect to $u$, then $x$ is not adjacent to $x_k$ in $G$ due to the outerplanarity, a contradiction. Similarly, if $x$ lies after $x_3$,  $x$ is not adjacent to $x_1$ in $G$ and we obtain the same contradiction. 

Now suppose that $x$ does not belong to $C$. Then, as every vertex in $X$ is adjacent to $x$ in $G$, we find that $x$ is at distance at most~$2$ in $H$ from $x_1$ and~$x_k$. It follows that there is a $(x_1,x_k)$-path $P$ in $H$ of length at most~$4$,
such that $P$, together with the edges $ux_1$ and $ux_k$, forms a cycle in $H$. The innerface of this cycle contain the
nonempty set $\{x_2,\ldots,x_{k-1}\}$, which is not possible as $G$ is outerplanar 
(alternatively, by contracting the edges of the subpath of $C$ from $x_2$ to $x_{k-1}$ and by contracting all but two edges of~$P$, we obtain a $K_{2,3}$, contradicting Lemma~\ref{lem:out-minor}).
\end{proof}

By combining Lemmas~\ref{lem:cons} and \ref{lem:size} we obtain the following lemma. 

\begin{lemma}\label{lem:del}
Let $H$ be a minimal outerplanar root of a graph $G$, and let $u\in V_G$. Then the following two statements hold:
\begin{itemize}
\item[(i)] If $x,y\in N_H(u)$ do not belong to the same block of $H$, then for any $X\in S(H,u)$, at least one of $x,y$ does not belong to $X$.
\item[(ii)] Let $F$ be a block of $H$ containing $u$ and vertices $x_1,\ldots,x_k\in N_H(u)$ ordered in clockwise order with respect to $u$ in the boundary cycle of $F$. 
Then, for any $X\in S(H,u)$, at least one of $x_i,x_j$ does not belong to $X$ if 
$|i-j|\geq 4$.
\end{itemize}
\end{lemma}

We now prove some structural results that help us to decide whether an edge incident to a vertex is in an outerplanar root of a graph or not.
Suppose that $X$ is a set of vertices of a graph~$G$ that are pairwise true twins, such that at least one vertex~$x\in X$ is a pendant vertex of a root~$H$ of $G$. Then in $G$, we find that $x$, and consequently, $X$ is simplicial. We therefore formulate the following lemma in terms of simplicial sets, although we do not need this fact for our proof.

\begin{lemma}\label{lem:pend} 
Let $H$ be a minimal outerplanar root of a graph $G$. If $G$ contains a set~$X$ of seven simplicial vertices that are pairwise true twins
in $G$, 
then at least one of the vertices in $X$ is a pendant vertex of~$H$.
\end{lemma}

\begin{proof}
First suppose that $X$ contains two vertices $x$ and $y$ that do not belong to the same block of $H$.  We claim that $x$ is a pendant vertex of $H$.   
Since $x$ and $y$ are adjacent in $G$, we find that $xu,yu\in E_H$ for a cut vertex $u$ that belongs to two blocks $F_x$ and $F_y$ of $H$ containing $x$ and $y$ respectively. 
To obtain a contradiction, assume that $x$ has a neighbour $z\neq u$ in $H$. Then $z$ is not in $F_y$. It follows that $zu\in E_H$, because $x$ and $y$ are true twins of $G$ and, therefore, $z\in N_G(y)$. By Lemma~\ref{obs:triangle}, $x$ or~$z$ has a neighbour $z'$ in $H$ with $z\notin N_H(u)$. In both cases, we find that $z'\in N_G(x)$, whereas $z\notin N_H(u)$ implies that $z\notin N_G(y)$. This is a contradiction to our assumption that~$x$ and~$y$ are true twins in~$G$.

Now suppose that all vertices of $X$ belong to the same block $F$ of $H$. Let $C$ be the boundary cycle of $F$. To choose some order, we 
let $x_1,\ldots,x_7$ be the vertices of $X$ numbered 
according to the clockwise order with respect to an arbitrary vertex of $C$. Because the vertices of $X$ are pairwise adjacent in $G$, $F$ has a chord $uv$ 
(where $u,v$ in $X$ is possible),
such that $X$ has vertices in both connected components of $F-\{u,v\}$. Among all such chords we choose $uv$ and a connected component $F'$ of $F-\{u,v\}$ in such a way that $F'$ contains the smallest number of vertices of $X$. Assume without loss of generality that $x_1\in V_{F'}$ and let $x_i,\ldots,x_j$ for some $1<i\leq j\leq 7$ be the vertices of $X$ in the other connected component $F''$ of $F-\{u,v\}$. Notice that $F''$ contains at least three vertices of $X$ by the choice of $uv$. Assume also that $v$ is after $x_1$ in the clockwise ordering  of the vertices of $C$ with respect to $u$.
Because the vertices of $X$ are adjacent in $G$, they are at distance at most~$2$ in $H$. As $F$ is outerplanar, this implies that for any $x\in X\cap V_{F'}$ and any $y\in X\cap V_{F''}$, $xu,uy\in E_H$ or $xv,vy\in E_H$.  

Suppose that $F'$ contains at least two vertices of $X$. By symmetry, we may assume that $x_1,x_2\in V_{F'}$. 
From our choice of $uv$, it follows that $x_1v, \notin E_H$ and $x_2u\notin E_H$, and thus $x_2v\in E_H$ and $x_1u\in E_H$. We find that $x_iu\in E_H$ and $x_jv\in E_H$, but these are intersecting chords of $C$; a contradiction. We conclude that $x_1$ is the unique vertex of $X$ in $F'$. 
This implies that $i\leq 3$ and $j\geq 6$, as it is possible that $u=x_7$ or $v=x_2$. 
We also deduce that $x_1u\in E_H$ or $x_1v\in E_H$. By symmetry we may  assume that $x_1u\in E_H$.

We now show that we may assume without loss of generality that~$u$ is adjacent to $x_1,\ldots,x_6$ in~$H$.
First suppose that $x_1v\notin E_H$. Then $x_hu\in E_H$ for every $h\in\{i,\ldots,j\}$. Hence, $u$ is adjacent to $x_1,\ldots,x_6$. 
Now suppose that $x_1v\in E_H$. We first show that $x_hu\in E_H$ for every $h\in\{i,\ldots,j\}$ or $x_hv\in E_H$ for every $h\in\{i,\ldots,j\}$. For contradiction, assume that there exists an index 
$s\in\{i,\ldots,j\}$ such that $x_su\notin E_H$ and an index $t\in\{i,\ldots,j\}$ such that $x_tv\notin E_H$. 
Since $x_su\notin E_H$, we find that $x_sv\in E_H$. As we cannot have intersecting chords in $C$, this means that $x_iu\notin E_H$ even if $i\neq s$. By using the same arguments with respect to index~$t$, we obtain $x_jv\notin E_H$ even if $j\neq t$. Because for $k\in \{i+1,j-1\}$, we have that $x_ku\in E_H$ or $x_kv\in E_H$ and $j-i\geq 3$, the distance between $x_i$ and $x_j$ in $H$ is at least~3 by the outerplanarity of $H$. As $x_i$ and $x_j$ are adjacent in $G$, we obtain a contradiction.  
Therefore, the claim holds. By symmetry, we may assume that $x_hu\in E_H$ for every $h\in\{i,\ldots,j\}$. Hence $u$ is adjacent to $x_1,\ldots,x_6$.

Let $y$ be the neighbour of $x_3$ on $C$ after $x_3$ in the clockwise order with respect to $u$. Note that $y\neq u$.
Because $x_3$ and $x_6$ are true twins of $G$, we have that $yx_6\in E_G$.
As $H$ is outerplanar 
and $u$ is adjacent to $x_4$ and $x_5$ in $H$, 
this means that 
$yu\in E_H$ and $yx_6\notin E_H$. 
We have that $uy,ux_3,yx_3\in E_H$. By Lemma~\ref{obs:triangle}, there is a vertex $z\neq u$ such that  either

\medskip
\noindent
 i) $x_3z\in E_H$ and $uz,yz\notin E_H$, or\\
 ii) $yz\in E_H$ and $uz,x_3z\notin E_H$.  

\medskip
\noindent
If $x_3z\in E_H$, then by the same arguments as for $y$, we find that $zu\in E_H$; a contradiction. 
Hence, we have that $yz\in E_H$. If $z\in \{x_1,\ldots,x_5\}$, then $zu\in E_H$; a contradiction.   
Therefore,   $z\notin \{x_1,\ldots,x_5\}$. Since $z$ is adjacent to $x_3$ in $G$, $z$ is a neighbour of $x_6$ in $G$.
As $H$ is outerplanar, the only possibility is that $x_4=y$ and moreover that $z$ lies on $C$ in between $x_4$ and $x_5$ and that $z$ is adjacent to $x_5$.
As $x_1$ and $x_5$ are true twins in $G$, we find that $z$ is adjacent to $x_1$ in $G$.
This means that $uz\in E_H$; a contradiction.  
We conclude that $x_3$ is a pendant vertex 
of~$H$. This completes the proof of the lemma.
\end{proof}

For the next lemmas, we need the following statement.

\begin{lemma}\label{lem:dist-cut}
Let $H$ be a square root of a graph $G$ and let $u$ be a cut vertex of $H$. Then for 
every $x,y\notin N_H[u]$ that are in distinct components of $H$, $\dist_{G-u}(x,y)\geq 3$.
\end{lemma}

\begin{proof}
Consider a shortest $(x,y)$-path $P$ in $G-u$. Then $P$ contains at least one edge $x'y'$ such that $x'$ and $y'$ are in distinct components of $H-u$, because $u$ is a cut vertex of $H$ and $x,y$ are in distinct components of $H-u$. Clearly, $x'y'\notin E_H$ and, therefore,  $x'u,y'u\in E_H$, as $u$ is a cut vertex of $H$. Hence, $x',y'\in N_H(u)$. Since $x,y\notin N_H[u]$, the vertices $x'$ and $y'$ are pairwise distinct from $x$ and $y$. This means that neither $x'$ nor $y'$ is an end-vertex of $P$. We conclude that $P$ has length at least~$3$, that is,  $\dist_{G-u}(x,y)\geq 3$.
\end{proof}

We need the following two lemmas in order to be able to identify the edges incident to a vertex of sufficiently high degree in an outerplanar root.

\begin{lemma}\label{lem:three}
Let $G$ be a graph with a minimal outerplanar root $H$. Let $u\in V_G$ be such that there are three distinct vertices $v_1,v_2,v_3\in N_G(u)$ that are pairwise at distance at least~$3$ in $G-u$. Then for every $x\in N_G(u)$, it holds that 
$xu\in E_H$ if and only if  $\dist_{G-u}(x,v_i)\leq 2$ for every $i\in \{1,2,3\}$.     
\end{lemma}

\begin{proof}
Let  $x\in N_G(u)$.

Observe that if there is some $i\in\{1,2,3\}$ such that $\dist_{G-u}(x,v_i)\geq 3$, then $xu\notin E_H$ by Lemma~\ref{lem:not-incl}. 

Suppose that $xu\notin E_H$, that is, $x\notin N_H[u]$.
As $v_1,v_2,v_3$ are neighbours of $u$ in $G$ that are pairwise at distance at least~$3$ in $G-u$, it follows
from Lemma~\ref{lem:not-incl} that $uv_i\notin E_H$ for every $i\in\{1,2,3\}$.
We must show that there is some $i\in\{1,2,3\}$ such that $\dist_{G-u}(x,v_i)\geq 3$.  
Observe that this property trivially holds if $x=v_h$ for $h\in \{1,2,3\}$. 
Assume that $x\notin \{v_1,v_2,v_3\}$. 
From Lemma~\ref{lem:not-incl} it follows that $uv_i\notin E_H$ for every $i\in\{1,2,3\}$, 
that is, $v_1,v_2,v_3\notin N_H[u]$.

Assume that there is an index $i\in\{1,2,3\}$ such that $x$ and $v_i$ are in distinct connected components of $H-u$. Then 
from Lemma~\ref{lem:dist-cut}, it follows that  $\dist_{G-u}(x,v_i)\geq 3$.

\begin{figure}[ht]
\centering
\scalebox{1}{
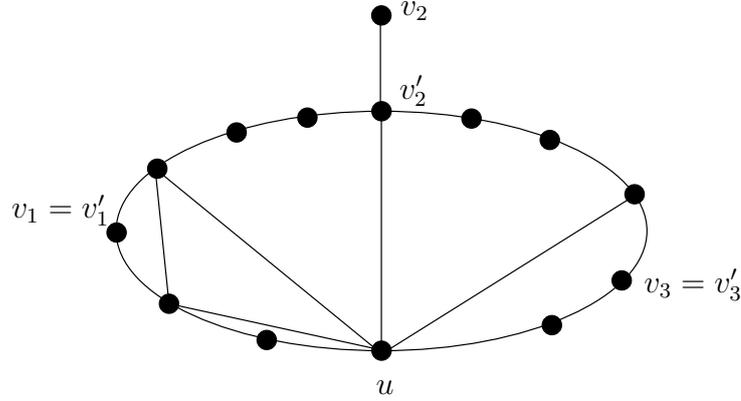}
\caption{When $v_1,v_2,v_3$ are in the same connected component of $H-u$.
\label{fig:case-one-comp}}
\end{figure}

Now suppose that $v_1,v_2,v_3$ and $x$ are in the same connected component of $H-u$.  Since $xu\in E_G$, there is a 
vertex $y\in N_H(u)$ such that $xy\in E_H$. Let $F$ be
the
block of $H$ containing $u$ and $y$ and let $C$ be the boundary cycle of $F$. Because $v_1,v_2,v_3$ and $x$ are in the same connected component of $H-u$ 
and are neighbours of $u$ in $G$, 
each $v_i$ is either 

\medskip
\noindent
i) a vertex of $F$ and we let $v_i'=v_i$ in this case, or\\[2pt]
ii) $v_i\notin V_F$ and there is a unique $v_i'\in V_{F}$ such that $v_iv_i'\in E_H$ and $v_i'u\in E_H$\\
\hspace*{3mm} (see Figure~\ref{fig:case-one-comp} for an example). 

\medskip
\noindent
Assume that $v_1',v_2',v_3'$ are in clockwise order with respect to $u$ in~$C$. 
Let $L_1$ and $L_2$ denote the $(v_1',v_2')$ and $(v_2',v_3')$-paths in $C$ avoiding $u$, respectively.
As $v_1,v_2,v_3$ are at distance at least~$3$ from each other in $G-u$, we observe that $L_1$ and $L_2$ have length at least 3.
Moreover, if $v_1=v_1'$ or $v_2=v_2'$, then the length of $L_1$ is at least 4, and if $v_2=v_2'$ or $v_3=v_3'$, then the length of $L_2$ is at least 4.

Observe that for $i\in\{1,2,3\}$, a shortest $(x,v_i)$-path is $G-u$ is a shortest path in $P_i^2-u$ for some $(x,v_i)$-path $P_i$ in $H$. 
Suppose that there is $i\in\{1,2,3\}$ such that $P_i$ contains $u$. Then every shortest $(x,v_i)$-path $P'$ in $P_i^2-u$ contains an edge $vw$ for $v,w\in N_H(u)$. Since $x,v_i\notin N_H[u]$, we obtain that neither $v$ no $w$ in an end-vertex of $P'$. Therefore, $P'$ has length at least 3 and, therefore, $\dist_{G-u}(x,v_i)\geq 3$. Assume from now that $P_1,P_2,P_3$ do not contain $u$.  Observe that it is sufficient to show that one of these paths has length at least 5. Observe also that by outerplanarity, the inner vertices of each $P_i$ form a segment of $C$ avoiding $u$.

First suppose that $x$ is a vertex of $C$. Then $x\neq v_2'\in N_H(u)$ and, therefore, $x$
lies either before $v_2'$ or after $v_2'$ in  the clockwise order with respect to $u$. 
Assume without loss of generality that $x$ is before $v_2'$. Then $P_3$ contains $L_2$ and at least one edge that is before $v_2'$. 
If $v_3\neq v_3'$, then $P_3$ also contains
the edge~$v_3'v_3$.
We conclude that the length of $P_3$ is at least 5. Hence, $\dist_{G-u}(x,v_3)\geq 3$.

Now suppose that $x$ does not lie on $C$. Then $y$ is a cut vertex of $H$ and it holds that $y$ and $u$ are in different components of $G-y$.
First suppose that $y\neq v_2'$, that is, $y$ lies on $C$ either before $v_2'$ or after $v_2'$. By symmetry, we can assume that $y$ is before $v_2'$. 
Then $P_3$ contains $L_2$ and at least two other edges. We obtain that    the length of $P_3$ is at least 5 and $\dist_{G-u}(x,v_3)\geq 3$.
Now suppose that $y=v_2'$. Then $P_3$ contains $L_2$ and $xv_2'$. If $v_3\neq v_3'$, then $P_3$ also contains $v_3'v_3$. We have that the length of $P_3$ is at least 5 and therefore, $\dist_{G-u}(x,v_3)\geq 3$.
\end{proof}

\begin{lemma}\label{lem:high} 
Let $G$ be a graph with a minimal outerplanar root $H$ such that any vertex has at most seven pendant neighbours in $H$. Let $u$ be a vertex with at least 22 neighbours in $H$. 
Then there are three distinct vertices $v_1,v_2,v_3\in N_G(u)$ that are pairwise at distance at least~$3$ in $G-u$. 
\end{lemma}

\begin{proof}
Let $P$ be the set of pendant neighbours of $u$ in $H$. Let $X=N_H(u)\setminus P$. Notice that $|X|\geq 15$, because $|P|\leq 7$.

First suppose that the vertices of $X$ belong to at least three connected components of $H-u$. Then there are three distinct blocks $F_1$, $F_2$ and $F_3$ of $H$ containing $u$ and at least~one vertex of $X$ each. 
By Lemma~\ref{lem:adj}, there are vertices 
$v_1,v_2,v_3\in N_H^2(u)$ 
such that $v_i$ is adjacent to a vertex of 
$V_{F_i}\cap X$
in $G$ for $i\in\{1,2,3\}$. Recall that $v_1$, $v_2$ and $v_3$ are in distinct connected components of $H-u$. We find that 
$v_1$, $v_2$, $v_3$ are pairwise at distance at least~3 in $G-u$ 
by Lemma~\ref{lem:dist-cut}.

Now suppose that the vertices of $X$ belong to exactly two connected components of $H-u$. Then there are two blocks $F_1$ and $F_2$ of $H$ containing $u$ and at least~one vertex of $X$ each. Since $|X|\geq 15$, we can assume that $F_1$ contains at least eight vertices of $X$, which we denote by $x_1,\ldots,x_k$ for $k\geq 8$ in clockwise order in the boundary cycle of $F_1$ with respect to $u$. By Lemma~\ref{lem:adj}, there are vertices 
$v_1,v_2\in N_H^2(u)$
such that $v_1$ is adjacent to $x_1$ in $G$ and $v_2$ is adjacent to $x_k$ in $G$. 
By Lemma~\ref{lem:del}~(ii), we find that
$v_1$ is not adjacent to $x_5,\ldots,x_k$ in $G$ and that $v_k$ is not adjacent to $x_1,\ldots,x_{k-4}$ in $G$. 
We observe that $v_1$ is either lying on the boundary cycle of $F_1$ or belongs to some other block of $H$ containing $x_1$ or $x_2$. Similarly, $v_2$ is either lying on the boundary cycle of $F_1$ or belongs to some other block of $H$ containing $x_{k-1}$ or $x_k$, respectively. Then $\dist_{G-u}(v_1,v_2)\geq 3$ 
(distance~3 is possible if $v_1$ lies on the boundary cycle between $u$ and $x_1$, and $v_2$ lies on the boundary cycle between
$x_k$ and $u$).
By Lemma~\ref{lem:adj}, there exists a vertex
$v_3\in N_H^2(u)$ 
such that $v_3$ is adjacent to a vertex of $F_2$ in $G$. 

Then $v_3$ is in a connected component of $H-u$ distinct from the connected component of $H-u$ to which $v_1$ and $v_2$ belong.
Hence $v_3$ is  at distance at least~3 from $v_1$ and $v_2$ in $G-u$ by Lemma~\ref{lem:dist-cut}.

Finally suppose that the vertices of $X$ all belong to the same connected component of $H-u$. That is, all
vertices of $X$ are in the same block $F$ of $H$, which also contains $u$. Denote them by $x_1,\ldots,x_k$ in their order in the clockwise order in the boundary cycle of $F$ with respect to $u$. By Lemma~\ref{lem:adj}, there exist vertices 
$v_1,v_2,v_3\in N_H^2(u)$ 
such that $v_1$ is adjacent to $x_1$ in $G$, $v_2$ is adjacent to $x_8$ and $v_3$ is adjacent to $x_k$. 
By Lemma~\ref{lem:del}~(ii), we find that
$v_1$ is not adjacent in $G$ to $x_5,\ldots,x_k$; $v_2$ is not adjacent to $x_1,\ldots,x_4$ and $x_{k-3},\ldots,x_k$; and $v_k$ is not adjacent to $x_1,\ldots,x_{k-4}$. Each of $v_1$, $v_2$ and $v_2$ is either lying on the boundary cycle of $F$ or is in another block of $H$ containing $x_1$ or $x_2$; $x_7$ or $x_8$ or $x_9$; or $x_{k-1}$ or $x_k$, respectively. This means that
$v_1,v_2,v_3\in N_G(u)$ are pairwise at distance at least~$3$ in $G-u$. 
\end{proof}

The next and final lemma of Section~\ref{sec:tech} will be crucial for our algorithm. In order to state it, we need to introduce some additional notation. Let $H$ be a minimal outerplanar root of a graph $G$, such that each vertex of $H$ is adjacent to at most seven pendant vertices. Let $U$ be the set of vertices that have degree at least~$22$ in~$H$.
 For every $u\in U$ and every block~$F$ of~$H$ containing~$u$, we 
consider the set  $X=N_H(u)\cap V_F$ and denote the vertices of $X$ by $x_1,\ldots,x_k$, where these vertices are numbered in the clockwise order with respect to $u$ in the boundary cycle of $F$.  Then we modify $G$ as follows:
\begin{itemize}
\item for  $i,j\in\{1,\ldots,k\}$ with $|i-j|\geq 4$, delete the edge $x_ix_j$ from $G$ (note that this edge exists in $G$);
\item for $i\in \{1,\ldots,k\}$ and  $y\in N_H(u)\setminus V_F$, delete the edge $x_iy$ from $G$ (note that this edge exists in $G$).
\end{itemize}
We denote the resulting graph by $G_H$; observe that $G_H$ is a spanning subgraph of $G$. 

In our final structural lemma we prove that $\tw(G_H)\leq 3\cdot 42^3$.

\begin{lemma}\label{lem:bound-tw} 
Let $G$ be a graph with a minimal outerplanar root~$H$, such that each vertex of $H$ is adjacent to at most seven pendant vertices. 
Then  $\tw(G_H)\leq 3\cdot 42^3$.
\end{lemma}

\begin{proof}
We first do the following for each vertex $u\in U$ 
(see Figure~\ref{fig:modif} for an example):
\begin{itemize}
\item Let $F_1,\ldots,F_r$ be the blocks of $H$ containing $u$.
\item For each $i\in\{1,\ldots,r\}$, denote by $x_1^i,\ldots,x_{k_i}^i$ the neighbours of $u$ in $F_i$ numbered according to the clockwise ordering with respect to $u$ in the boundary cycle of $F_i$. Assume that $x_0,\ldots,x_k$ is the ordering of $N_H(u)$ obtained by the consecutive concatenation of the sequences $x_1^i,\ldots,x_{k_i}^i$ for $i=1,\ldots,r$.
\item Modify $H$ as follows: delete $u$ from $H$ and add  a path $u_1\ldots u_k$ such that $u_i$ adjacent to $x_{i-1}$ and $x_i$ for $i=1,\ldots,k$.
\end{itemize}
Let $\hat{H}$ be the graph obtained by the above procedure. 
Note that the procedure modifies vertices and degrees of vertices of $H$.

\begin{figure}[ht]
\centering
\scalebox{0.78}{
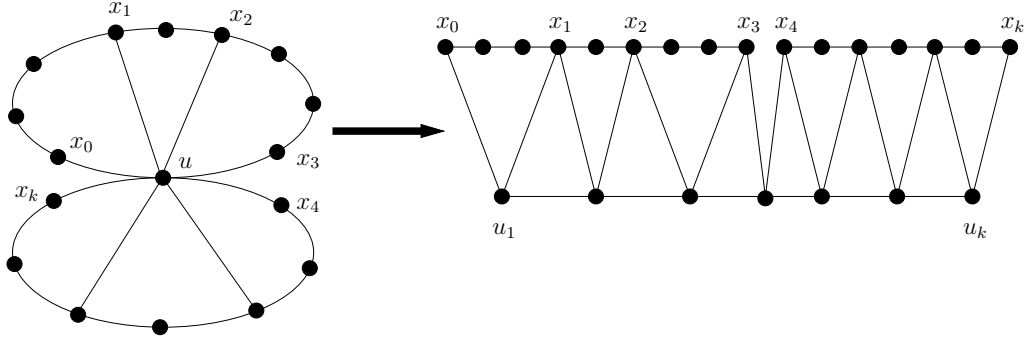}
\caption{An example of the modification of $H$ for one vertex~$u$, where $r=2$. 
If there were three blocks, the third block would have been placed after the second block and so on.
\label{fig:modif}}
\end{figure}

We observe that $H$ is a contraction of $\hat{H}$, as $H$ can be obtained from $\hat{H}$ by contracting the path $u_1\ldots u_k$ constructed for each $u\in U$. Moreover, $\hat{H}$ is an outerplanar graph, because each step maintains outerplanarity (see Figure~\ref{fig:modif}). By Lemma~\ref{obs:tw-op}, we find that $\tw(\hat{H})\leq 2$.

We are first going to prove that $\Delta(\hat{H})\leq 42$. Let $v\in V_{\hat{H}}$. Suppose first that $v\in V_H\setminus U$. We have that $d_H(v)\leq 21$ and, in particular, $v$ has at most~$21$ neighbours in $U$ in the graph $H$. In the construction of $\hat{H}$, each neighbour of this type is replaced by two neighbours and all other neighbours remain the same. Therefore, $d_{\hat{H}}(v)\leq 42$. Suppose now that $v$ is a vertex of one of the paths $u_1\ldots u_k$ constructed for $u\in U$. When $u$ is replaced by $u_1\ldots u_k$,  the degree of each vertex $u_i$ is at most~$4$, and at most two neighbours of $u_i$ are modified in the subsequent construction steps. This implies that $d_{\hat{H}}(v)\leq 6$.

We are now going to prove that $G_H$ is a minor of $\hat{H}^4$. 
Let $\hat{G}$ be the graph obtained from  $\hat{H}^4$ after contracting each constructed path $u_1\ldots u_k$ into a single vertex, which we denote by $u$ again. Hence $V_{\hat{G}}=V_G$. We show that $G_H$ is a subgraph of $\hat{G}$. 

We already observed that $H$ can be obtained from $\hat{H}$ by contracting paths $u_1\ldots u_k$ constructed for $u\in U$. Hence, each edge 
of $G_H$ that is an edge
of $H$ is an edge of $\hat{G}$. Let $xy$ be an edge of $G_H$ that is not an edge of $H$. Then there is 
a vertex $u\in V_G$ such that $xu,yu\in E_H$. Denote by $X'$ and~$Y'$, respectively, the sets of vertices of $\hat{H}$ that are contracted to $x$ and $y$ 
in $\hat{G}$, 
respectively. If $u\notin U$, then by the construction of $\hat{H}$, there are 
vertices $x'\in X'$ and $y'\in Y'$ such that $x'u,y'u\in E_{\hat{H}}$. Hence, $x'y'\in \hat{H}^4$ and thus $xy\in E_{\hat{G}}$. Suppose that $u\in U$. By the definition of $G_H$, the vertices $x$ and $y$ are in the same block $F$ of $H$. Denote by $z_1,\ldots,z_k$ the vertices of $N_H(u)$ in $F$ in the clockwise order with respect to $u$ along the boundary cycle of $F$. We have that $x=z_i$ and $y=z_j$ for some $i,j\in \{1,\ldots,k\}$. By the definition of $G_H$, 
$|i-j|\leq 3$. 
By the construction of $\hat{H}$, there are 
vertices $x'\in X'$ and $y'\in Y'$ that are 
joined by the path 
$x'u_{i+1}\ldots u_jy'$ in $\hat{H}$.
Since this path has length at most~$4$, 
we find that $x'y'\in \hat{H}^4$, and therefore, $xy\in E_{\hat{G}}$.  

Since $G_H$ is a subgraph of $\hat{G}$ and $\hat{G}$ is a contraction of $\hat{H}^4$, we conclude that $G_H$ is a minor of $\hat{H}^4$. 
Since $G_H$ is a minor of $\hat{H}^4$, we find that $\tw(G_H)\leq \tw(\hat{H}^4)$ by Lemma~\ref{obs:minor}. Because $\hat{H}$ is outerplanar, $\tw(\hat{H})\leq 2$ by Lemma~\ref{obs:tw-op}, and because $\Delta(\hat{H})\leq 42$, $\tw(\hat{H}^4)\leq (\tw(\hat{H})+1)\cdot 42^3$ by Lemma~\ref{lem:tw-power}. Hence, $\tw(G_H)\leq 3\cdot 42^3$. 
\end{proof}

\subsection{The Algorithm}\label{sec:alg}
In this section, we construct our $O(n^4)$-time algorithm for \textsc{Outerplanar Root}, that is, we are now ready to prove Theorem~\ref{thm:outerplanar}.

\medskip
\noindent
{\bf  Theorem~\ref{thm:outerplanar} (restated).}
{\it \textsc{Outerplanar Root} can be solved in $O(n^4)$ time.}

\begin{proof}
Let $G$ be the input graph.  
We may assume without loss of generality that $G$ is connected and has $n\geq 2$ vertices. 
We first exhaustively apply the following rule in order to
reduce the number of pendant vertices adjacent to the same vertex in a (potential) outerplanar root of $G$.

\medskip
\noindent
{\bf Deleting a simplicial true twin.} 
If $G$ has a set $X$ of simplicial true twins of size at least~8, then delete an arbitrary
vertex $u\in X$ from $G$.

\smallskip
\noindent
The following claim shows that this rule is safe.

\medskip
\noindent
{\bf Claim 1.}
{\it If $G'=G-u$ is obtained from $G$ by the application of {\bf deleting a simplicial true twin}, then $G$ has an outerplanar root if and only if $G'$ has an outerplanar root.}

\medskip
\noindent
We prove Claim~1 as follows.
First suppose that $G$ has an outerplanar root~$H$, which we may assume to be minimal.
By Lemma~\ref{lem:pend}, $H$ has a pendant vertex $u\in X$. It is readily seen that $H'=H-u$ is an outerplanar root of $G'$. 
Now suppose that $G'$ has an outerplanar root~$H'$, which we may assume to be minimal. 
By  Lemma~\ref{lem:pend}, $H'$ has a pendant vertex $w\in X\setminus \{u\}$, since the vertices of $X\setminus\{u\}$ are simplicial true twins of $G'$ and  $|X\setminus\{u\}|\geq 7$. Let $v$ be the unique neighbour of $w$ in $H'$. We construct $H$ from $H'$ by adding $u$ and making $u$ adjacent to $v$. It is readily seen that $H$ is an outerplanar root of $G$. 
This proves Claim~1.

\medskip
\noindent
For simplicity, we call the graph obtained by the exhaustive application of {\bf deleting a simplicial true twin} $G$ again. 
The next claim immediately follows from the observation that any two pendant vertices of a square root $H$ of $G$ adjacent to the same vertex in $H$ are simplicial true twins of $G$. 

\medskip
\noindent
{\bf Claim~2.} {\it Every outerplanar root of $G$ has at most seven pendant vertices adjacent to the same vertex.}

\medskip
\noindent
In the next stage of our algorithm we are going to label some edges of $G$ \emph{red} or \emph{blue} in such a way that the red edges are included in every minimal outerplanar root of $G$, whereas the  blue edges are excluded from any minimal outerplanar root of $G$. Let $R$ be the set of red edges and $B$ be the set of blue edges. We will also construct a set of vertices~$U$ 
of~$G$ such that for every $u\in U$, 
all edges incident to $u$ are labeled red or blue. 

\medskip
\noindent
{\bf Labeling edges.} Set $U=\emptyset$, $R=\emptyset$ and $B=\emptyset$. For each $u\in V_G$ such that there are three distinct vertices $v_1,v_2,v_3\in N_G(u)$ that are at distance at least~$3$ from each other in $G-u$, do the following:
\begin{itemize}
\item[(i)] set $U=U\cup\{u\}$;
\item[(ii)] set $B'=\{ux\in E_G\mid \text{there is an } 1\leq i\leq 3\text{ such that }\dist_{G-u}(x,v_i)\geq 3\}$;
\item[(iii)] set $R'=\{ux\mid x\in N_G(u)\}\setminus B'$;
\item[(iv)] set $R=R\cup R'$ and $B=B\cup B'$;
\item[(v)] if $R\cap B\neq\emptyset$, then return a no-answer and stop.
\end{itemize}
\noindent
Note that the above rule does not change the graph $G$ itself. 
Lemmas~\ref{lem:three} and \ref{lem:high}, combined with Claim~2, imply the following claim.

\medskip
\noindent
{\bf Claim~3.}
{\it If $G$ has a minimal outerplanar root~$H$, then
{\bf labeling edges} does not stop in step~{\rm (v)}. Moreover, $R\subseteq E_H$ and $B\cap E_H=\emptyset$, and every vertex $u\in V_G$ with $d_H(u)\geq 22$ is included in $U$.}

\medskip
\noindent
Next, we are going to find, for each $u\in U$, a set~$S$ of edges $xy$ with $xu,yu\in R$ that may be removed from~$G$.
This way we will reduce the treewidth of $G$.

\medskip
\noindent
{\bf Deleting irrelevant edges.} Set $S=\emptyset$. For every vertex $u\in U$ and every pair of distinct vertices $x,y\in N_G(u)$ such that $xu,uy\in R$ do the following:
\begin{itemize}
\item[(i)] if $xy\notin E_G$, then return a no-answer and stop;
\item[(ii)] if
 there is no  $v\in N_G(u)$ such that $vu\in B$ and $x,y\in N_G(v)$,
then include $xy$ in $S$;
\item[(iii)] if $R\cap S\neq\emptyset$, then return a no-answer and stop; 
\item [(iv)] remove the edges of $S$ from $G$.
\end{itemize}

\noindent
By combining Lemma~\ref{lem:non-edges} with Claim~3 we obtain the following claim.

\medskip
\noindent
{\bf Claim 4.}
{\it If $G$ has a minimal outerplanar root~$H$, then {\bf deleting irrelevant edges} does not stop in step~{\rm (i)} or~{\rm (iii)}, and moreover,  $S\cap E_H=\emptyset$.}

\medskip
\noindent
Assume that we have not returned a no-answer after the execution of {\bf deleting irrelevant edges}.
Let $G'=G-S$. Because of the edge deletions, a square root of $G$ may not be a square root of $G'$ and vice versa. Nevertheless, the edge labels and the properties of the edges of $S$ allow us to recover the structure of square roots of $G$ from $G'$. In order to show this, we prove the following claim.

\medskip
\noindent
{\bf Claim 5.} 
{\it The graph $G$ has an outerplanar root if and only if there is a set $L\subseteq E_{G'}$ such that 
\begin{itemize}
\item[(i)] $R\subseteq L$ and $B\cap L=\emptyset$;
\item[(ii)]  for every $xy\in E_{G'}$, $xy\in L$ or there exists  a vertex~$z\in V_{G'}$ with $xz,zy\in L$;
\item[(iii)] for every two distinct edges $xz,yz\in L$, 
$xy\in E_{G'}$ or there is a vertex $u\in U$ with $xu,uy\in R$; and
\item[(iv)] the graph $H=(V_G,L)$ is outerplanar.
\end{itemize}}

\noindent
We prove Claim~5 as follows.
First suppose that $H$ is a minimal outerplanar root of $G$. By Claim~4 we find that
$E_H\cap S=\emptyset$, that is, $E_H\subseteq E_{G'}$. 
Let $L=E_H$.
Then (i) holds due to Claim~3, whereas (ii) and (iv) hold because $H=(V_G,L)$ is an outerplanar root of $G$. 
To prove (iii) suppose that $xz$ and $zy$ are distinct edges of $L$ such that $xy\notin E_{G'}$.
As $H=(V_G,L)$ is a square root of $G$, this means that $xy\in E_G\setminus E_{G'}$, that is, $xy\in S$.
By definition of the rule {\bf deleting irrelevant edges}, this means that there must exist a vertex $u\in U$ such that $xu,uy\in R$.

Now suppose that there is a subset $L\subseteq E_{G'}$ such that (i)--(iv) hold.
Let $xy\in E_G$. If $xy\in E_{G'}$, then  $xy\in L$ or there is a vertex $z\in V_{G'}$ such that $xz,yz\in L$ by~(ii).
If $xy\in E_G\setminus E_{G'}=S$, then there is a vertex $u\in U$ such that $xu,uy\in R$ by 
(iii).
As $R\subseteq L$ by~(i), 
we find that $xu,uy\in L$. Hence $G$ is a subgraph of $(V_G,L)^2$. As $L\subseteq E_{G'}$, we find that $G=(V_G,L)^2$.
We conclude that $H=(V_G,L)$ is a square root of $G$. By (iv) we find that $H$ is an outerplanar root of $G$.
Hence we have proven Claim~5.

\medskip
\noindent
It remains to check the existence of a set of edges $L$ satisfying (i)--(iv) of Claim~5 for a given triple $G'$, $R$, $B$, which is the final step of the algorithm. Notice that, if $G$ has a minimal outerplanar root $H$, then $G'$ is a subgraph of the graph $G_H$ constructed in Section~\ref{sec:tech}; this is due to 
Lemmas~\ref{lem:non-edges} and~\ref{lem:del}. 
By Lemma~\ref{lem:bound-tw},  we have that
$\tw(G_H)\leq 3\cdot 42^3$. Hence we must return a no-answer and stop if $\tw(G')>3\cdot 42^3$.

Now suppose  $\tw(G')\leq 3\cdot 42^3$.
It is straightforward to verify that properties (i)--(iv) in Claim~5 can be expressed in MSO. In particular, to express outerplanarity in (iv), we combine  
Lemma~\ref{lem:out-minor} with Lemma~\ref{l-engel}. Afterwards we use Lemma~\ref{l-courcelle}.

The correctness of our algorithm follows from the above description and proofs of Claims~1--5.
It remains to evaluate the running time of our algorithm, which we do below.

\medskip
\noindent
It is well-known that the classes of true twins can be constructed in linear time (see, for example,~\cite{GHL}).
Then we can check whether each class contains simplicial vertices in $O(n^2)$ time.
Therefore, the exhaustive application of {\bf deleting a simplicial true twin} costs 
$O(n^2)$ time.
For every vertex $u$, we can compute the distances between the vertices of $N_G(u)$ in $G-u$ in $O(n^3)$ time. 
This implies that {\bf labeling edges} can be done in $O(n^4)$ time. 
Applying {\bf deleting irrelevant edges} takes $O(n^4)$ time as well, as it takes $O(n^2)$ to process a pair $x,y$ and the number of 
such pairs is $O(n^2)$. We construct $G'$  in linear time. Finally, checking whether $\tw(G')\leq 3\cdot 42^3$ and deciding whether there is a set of edges $L$ satisfying the required properties can be done in linear time by Lemma~\ref{l-bod} and~\ref{l-courcelle}, respectively. 
Hence the total running time is $O(n^4)$. This completes the proof of Theorem \ref{thm:outerplanar}.
\end{proof}

We conclude the section by the remark that instead of merely checking the existence of
a set $L$ as in Claim~5, we can also find $L$ if it exists. We can do this by constructing a dynamic programming algorithm for graphs of bounded treewidth 
(see~\cite{CochefertCGKP13} for a sketch of such an approach).
Hence, if $G$ has an outerplanar root, then we can find it in polynomial time.

\section{Roots of Pathwidth at Most~2}\label{sec:pw}
We say that a square root $H$ of $G$ is a \emph{pathwidth-$2$ root} if $H$ has pathwidth at most~2, and we define the following problem:

\problemdef{\sc Pathwidth-2 Root}{a graph $G$.}{does $G$ have a pathwidth-$2$ root?}

The main result of this section is the following theorem.

\begin{theorem}\label{thm:pw2}
\textsc{Pathwidth-2 Root} can be solved in $O(n^6)$ time.
\end{theorem}
\noindent
We first show a number of structural results in Section~\ref{secpw:strc}.
We then use these results in the design of our polynomial-time algorithm for {\sc Pathwidth-2 Root} in Section~\ref{secpw:algo}.

\subsection{Structural Lemmas} \label{secpw:strc}

Recall our assumption that $A\nsubseteq B$
for every two distinct bags $A$ and $B$ of a path decomposition.
The class of graphs of pathwidth at most~2 is closed under 
vertex deletion and edge deletion.
Hence, by Lemma~\ref{obs:min}, we may focus on minimal  pathwidth-2 roots.

A graph $H$ of pathwidth at most~2 may have several different path decompositions of width at most~2. We can use any such path decomposition in our arguments below. For ease of notation, we will refer to such a path decomposition as {\it the} path decomposition of $H$.

\begin{lemma} \label{lem:neighxi}
Let $H$ be a minimal pathwidth-$2$ root of a graph~$G$. If there are distinct vertices $u, v, x_1, \ldots, x_k$ such that 
the path decomposition of $H$ contains
bags $\{x_1, u, v\}$, $\{x_2, u, v\}$,$\ldots$, $\{x_k, u, v\}$ in this order, then $N_H(x_i)\subseteq \{u,v\}$ for $i=2,\ldots, k-1$. 
\end{lemma}

\begin{proof}
Suppose $x_i$ has a neighbour $w$ in $H$ such that $w\neq u$ and $w\neq v$. There exists a bag $B$ in the path decomposition of $H$ that contains $x_i$ and $w$. As $B$ contains $x_i$, we find that $B$ is between the bags $\{x_1,u,v\}$ and $\{x_k,u,v\}$ in the path decomposition and hence must contain $u$ and $v$. Then $|B|\geq 4$, a contradiction with $\pw(H)\leq 2$. 
\end{proof}

The \emph{Ramsey number} $R(p,q)$ is the smallest integer $n$ such that every graph on $n$ vertices has either a clique of size $p$ or an independent set of size $q$. 
By Ramsey's Theorem~\cite{Ra30}, $R(p,q)$ is finite for every pair of integers $p,q\geq 0$.
We use Ramsey's Theorem in the proof of the following lemma.

\begin{lemma} \label{lem:pwreduction1}
Let $H$ be a minimal pathwidth-$2$ root of a graph $G$.
Then there is a constant $c_1$ such that for every set $W$ of true twins in $G$ with $|W|\geq c_1$, one of the following holds:
\begin{itemize}
\item[(i)] $W$ contains a pendant vertex of $H$.
\item[(ii)] $W$ contains three pairwise nonadjacent vertices $x,y,z$ of degree~$2$ in $H$ with  $N_H(x)=N_H(y)=N_H(z)$.
\end{itemize}
\end{lemma}

\begin{proof}
Let $c_1=R(4,16)$ and consider a set~$W$ of true twins in $G$ with $|W|\geq c_1$.
We first construct an auxiliary graph $F$. Let $V_F=W$. We add an edge between two vertices of $F$ 
if and only if
there exists a bag in the path decomposition of $H$ that contains both of them. We claim that $F$ has an independent set of size~16. As $c_1=R(4,16)$, it suffices to 
prove that $F$ does not contain a $K_4$. For contradiction assume that $F$ has a $K_4$ with vertex set $\{x_1,x_2,x_3,x_4\}$. Let $P_i$ be the path formed by the bags containing vertex~$x_i$ in the path decomposition of $H$. 
As $\{x_1,x_2,x_3,x_4\}$ is a clique in $F$, any two paths $P_i$ and $P_j$ are intersecting.
By the Helly property, there exists a bag containing all four vertices, a contradiction with $\pw(H)\leq 2$. 
Hence, $F$ does not contain a $K_4$.

Let $W'=\{x_1,x_2,\ldots,x_{16}\}$ be an independent set of $F$ (so $W'\subseteq W$). By the construction of $F$, there are no two vertices of $W'$ that are contained in the same bag of the path decomposition of $H$. Let $B_1,B_2,\ldots, B_{16}$ be (distinct) bags 
that appear in this order in the path decomposition of $H$,
such that $x_i\in B_i$ for $i=1,\ldots,16$.
As the vertices of $W'$ are true twins in $G$ and they are not adjacent in $H$, there must exist a path of length~2 in $H$ between any two of them. Let $x_1ux_{16}$ be such a path in $H$ between $x_1$ and $x_{16}$. We may assume without loss of generality that $u\in B_1$ and $u\in B_{16}$, 
since there exists a bag that contains both $u$ and $x_1$ and a bag that contains both $u$ and $x_{16}$, and $B_1$ and $B_{16}$ can be chosen to be any bags containing $x_1$ and $x_{16}$, respectively. 
By definition,  $u\in B_i$ for $i=2,\ldots,15$. 

First assume that there are three distinct vertices $x_i$, $x_j$ and $x_k$ in $W'$ with $2\leq i<j<k\leq 15$ that are not adjacent to $u$ in $H$. Let $v$ be the vertex in a path of length~2 in~$H$ between $x_i$ and $x_k$. We may assume without loss of generality that $B_i=\{x_i,u,v\}$ and $B_k=\{x_k,u,v\}$. Then $v\in B_j$. By Lemma~\ref{lem:neighxi}, we
obtain $N_H(x_j)\subseteq \{u,v\}$. Then, as $x_ju\notin E_H$, we find that $N_H(x_j)=\{u\}$.
 Hence condition~(i) holds.

Now assume that at most two vertices of $W'\setminus \{x_1,x_{16}\}$ are not adjacent to $u$ in $H$. Let $W''=\{x_1',\ldots,x_p'\}$ 
consist of all vertices of $W'\setminus \{x_1,x_{16}\}$ that are adjacent to $u$ in $H$; note that $p\geq 12$ and that $W''$ might be a proper subset of $W'\setminus \{x_1,x_{16}\}$.
If some vertex of $W''$ has degree~1 in $H$, then condition~(i) holds. Suppose that all the vertices of $W''$ have degree at least~2 in $H$. 
Let $\mathcal{B'}=\{B_1',\ldots,B_p'\}$ be bags that appear in this order in the path decomposition of $H$
such that $x_i'\in B_i'$ for $i=1,\ldots,p$. 

Consider a neighbour $v_1\neq u$ of $x_1'$ in $H$. We may assume without loss of generality that $B_1'=\{x_1',u,v_1\}$. 
First suppose that  
$v_1$ appears in at least five bags of $\mathcal{B'}$. Then, by definition, $v_1$ must be in $B_1'$, $B_2'$, $B_3'$, $B_4'$ and $B_5'$.
As $x_2'$, $x_3'$ and $x_4'$ do not have degree~1 in $H$, we use Lemma~\ref{lem:neighxi} to find that 
$N_H(x_2')=N_H(x_3')=N_H(x_4')=\{u,v_1\}$.
Hence condition~(ii) holds.
From now on assume that no neighbour of $x_1'$ in $H$ appears in more than five bags of $\mathcal{B'}$, that is, any neighbour of $x_1'$ may only appear in $B_1',\ldots,B_4'$. 
This implies that in order to prove condition~(ii) it suffices to find a vertex $v_i$ that appears in at least five bags of $\mathcal{B}\cup B_{16}$.

Suppose that $H$ contains a path $x_1'v_1v_2x_{16}$ for some vertices $v_1,v_2$ with $u\notin \{v_1,v_2\}$. 
We may assume without loss of generality that $B_1'=\{x_1',u,v_1\}$ and $B_{16}=\{x_{16},u,v_2\}$. 
Recall that $v_1$ does not belong to any bag $B_i'$ for $i\geq 5$ and $x_{16}$ belongs to $B_{16}$, 
while $u$ belongs to any bag between $B_1$ and $B_{16}$. Then there exists a bag $\{v_1,v_2,u\}$, which  has to be between $B_1'=\{x_1',u,v_1\}$ and $B_{16}=\{x_{16},u,v_2\}$. Since $p\geq 12$ and $v_1$ does not belong to any $B_i'$ for $i\geq 5$, this means that
$v_2$ appears in at least five bags of $\mathcal{B'}$. 

Suppose that $H$ contains a path $x_1'v_1v_2v_3x_{16}$ for some vertices $v_1,v_2,v_3$ with $u\notin \{v_1,v_2,v_3\}$. 
We may assume without loss of generality that $B_1'=\{x_1',u,v_1\}$ and $B_{16}=\{x_{16},u,v_3\}$.
Recall that $v_1$ does not belong to any bag $B_i'$ for $i\geq 5$, 
while $u$ belongs to any bag between $B_1$ and $B_{16}$. Then there exists bags $\{v_1,v_2,u\}$ and $\{v_2,v_3,u\}$, which have to be between
$B_1'=\{x_1',u,v_1\}$ and $B_{16}=\{x_{16},u,v_3\}$. Since $p\geq 12$ and $v_1$ does not belong to any $B_i'$ for $i\geq 5$, this means that either
$v_2$ or $v_3$ appear in at least five bags of $\mathcal{B'}\cup B_{16}$. 

We continue as follows.
Let $v\neq u$ be a neighbour of $x_1'$. 
By the above assumption, $v$ does not belong to $B_i'$ for $i\geq 5$. In particular, this means that $v$ is not adjacent to $x_{16}$.
As $x_1'$ and $x_{16}$ are true twins in~$G$, we find that $v$ is a neighbour of $x_{16}$ in $G$. 
As  $v$ and $x_{16}$ are not adjacent in $H$, this means that $G$ contains a path $vwx_{16}$ for some vertex~$w$.
If $w\neq u$, then $H$ contains a path $x_1'vwx_{16}$ with $u\notin \{v,w\}$. Hence, condition~(ii) holds.
Now suppose that that $w=u$.
Then $v$, $u$ and $x_1'$ form a triangle in $H$. By Lemma~\ref{obs:triangle}, $H$ contains a vertex $z\neq u$ that is adjacent to at least one of $v$ or $x_1'$, but not to~$u$. 

First suppose $vz\notin E_H$. Then $x_1'z\in E_H$. As $x_1'$ and $x_{16}$ are true twins in $G$, we find that $z$ is also adjacent to $x_{16}$ in $G$. Hence $H$ either contains a path $x_1'zx_{16}$ or a path $x_1'zz'x_{16}$ for some vertex $z'$. Note that $u\notin \{z,z'\}$, as $z$ is neither equal to $u$ nor adjacent to $u$. Hence we find that condition~(ii) holds. 

Finally suppose $vz\in E_H$. 
As $\dist_H(z,x_1')\leq 2$, $z$ is adjacent to $x_1'$ in $G$.
As $x_1'$ and $x_{16}$ are true twins in $G$, we find that $z$ is also adjacent to $x_{16}$ in $G$. 
If $zx_{16}\in E_H$, then we have found a path $x_1'vzx_{16}$ with $u\notin \{v,z\}$ and thus condition~(ii) holds.
If $zx_{16}\notin E_H$, then $H$ contains a path $x_1'vzz'x_{16}$ for some vertex~$z'$. 
Note that $u\notin \{z,z'\}$, as $z$ is neither equal to $u$ nor adjacent to $u$. 
Hence, as $u\neq v$ either, condition~(ii) also holds in this case.    
\end{proof}

A graph $G$ that contains no set of more than $c_1$ vertices that are true twins of each other is called {\it $c_1$-twin-bounded}.

\begin{lemma} \label{lem:pwboundxi}
Let $G$ be a $c_1$-twin-bounded graph that has a minimal pathwidth-$2$ root~$H$. If there are distinct vertices $u, v, x_1, \ldots, x_k$ such that the bags $\{x_1, u, v\}$, $\{x_2, u, v\}$,$\ldots$, $\{x_k, u, v\}$ appear in this order in the path decomposition of $H$, then $k\leq 3c_1+2$. 
\end{lemma}

\begin{proof}
By Lemma~\ref{lem:neighxi}, we have $N_H(x_i)\subseteq \{u,v\}$ for $i=2,\ldots, k-1$. .
Vertices adjacent only to $u$ in~$H$ are true twins in $G$. The same applies for vertices only adjacent to~$v$ in~$H$ and to vertices only adjacent to $u$ and $v$ in~$H$. As the size of every set of true twins in $G$ is bounded by $c_1$, we obtain $|\{x_2,\ldots,x_{k-1}\}|\leq 3c_1$ and thus $k\leq 3c_1+2$. 
\end{proof} 

\begin{lemma} \label{lem:pwcommonneigh}
Let $G$ be a $c_1$-twin-bounded graph that has a minimal pathwidth-$2$ root~$H$. 
Any two vertices $u$ and $v$ have at most $c_1+2$ common neighbours in $H$.  
\end{lemma}

\begin{proof}
Let $N_H(u)\cap N_H(v)=\{x_1,\ldots,x_t\}$. 
If $t\leq 2$, then $t\leq c_1+2$. Suppose $t\geq 3$.
The path decomposition of $H$ must have a bag containing $u$ and $x_i$ and a bag containing $v$ and $x_i$ for $i=1,\ldots,t$. As $\pw(H)\leq 2$, this implies the existence of the bag $B_i=\{u,v,x_i\}$ for each~$1\leq i\leq t$. 
In order to see this, assume that the bags containing $\{x_1,u\}$, $\{x_2,u\}$, \ldots, $\{x_t,u\}$ appear in this order in the path decomposition of $H$.  For 
$2\leq i<j\leq t-1$, there is no bag containing both $x_i$ and $x_j$, since there must exist bags containing
$\{x_1,v\}$, $\{x_1,u\}$, $\{x_t,v\}$ and $\{x_t,u\}$.
Assume that the bag containing $\{x_1,v\}$ appears after the one containing $\{x_1,u\}$. Since there exists a bag containing $\{x_t,u\}$, there exists a bag  $B_1=\{x_1,u,v\}$. Now, for every $2\leq i\leq t-1$, the bag containing $\{x_i,v\}$ must also contain $u$, because of the existence of the bag containing $\{x_t,u\}$. 
Hence, for every $2\leq i\leq t-1$, there exists a bag $B_i=\{u,v,x_i\}$.
Finally, since there exist bags containing both $\{x_t,v\}$ and $\{x_t,u\}$ and a bag $\{x_{t-1},u,v\}$, we conclude that there is a bag $B_t=\{x_t,u,v\}$.
The above implies that we may also assume that $B_1,\ldots, B_{t}$ appear in the path decomposition of $H$ in this order. By Lemma~\ref{lem:neighxi}
we find that $N_H(x_i)\subseteq \{u,v\}$ for $i=2,\ldots t-2$. As each $x_i$ is adjacent to $u$ and $v$, this means that
$N_H(x_i)= \{u,v\}$ for $i=2,\ldots t-2$. Consequently, $x_2,x_3,\ldots,x_{t-1}$ are true twins in $G$. Hence, as $G$ is $c_1$-twin-bounded, $t-2\leq c_1$, and thus $t\leq c_1+2$. 
\end{proof}

\begin{lemma} \label{lem:pwfive}
Let $G$ be a $c_1$-twin-bounded graph that has a minimal pathwidth-$2$ root~$H$.  Let $c_2=6\cdot 21(c_1+2)$.
Let $u$ be a vertex with $d_H(u)\geq c_2$. Then there are five distinct vertices $x_1,\ldots, x_5\in N_G(u)$ that are pairwise
at distance at least~$3$ in $G-u$.
\end{lemma}  

\begin{proof}
Choose a set of bags $B_1,\ldots, B_l$ in the path decomposition of $H$, such that $u\in B_i$ for $i=1,\ldots,l$ and $N_H(u)\subseteq
\cup_{i=1}^lB_i$. Note that some neighbours of $u$ might appear in more than one bag of this set. 

Let $k_1$ be the smallest integer such that $\cup_{i=1}^{k_1}B_i$ contains at least three distinct vertices of $N_H(u)$. Since $u$ belongs to all bags and every bag has size at most~3, at least one of these three neighbours in $\cup_{i=1}^{k_1}B_i$ does not appear in $B_{k_1}$. Let $v_1$ be such vertex.
For $j\geq 2$,
let $k_j$ be the smallest integer greater than $k_{j-1}$ such that $\cup_{i=k_{j-1}}^{k_j}B_i$ contains at least five new vertices of $N_H(u)$. 
As $u$ belongs to all bags, there is at least one vertex~$v_j$ among these five vertices that appears neither in $B_{k_{j-1}}$ nor in $B_{k_j}$. This yields an independent set $\{v_1,\ldots,v_t\}\subset N_H(u)$. Since $d_H(u)\geq 6\cdot 21(c_1+2)$, we have $t\geq 21(c_1+2)$. 

Since $G$ is $c_1$-twin-bounded and vertices that have the same vertex as their unique neighbour in $H$ are true twins in $G$, at least $t-c_1$ vertices from $\{v_1,\ldots v_t\}$ have another neighbour in $H$ besides $u$.
By  Lemma~\ref{lem:pwcommonneigh}, two vertices can have at most $c_1+2$ common neighbours in $H$.
Hence, we can pick 21 vertices from $\{v_1,\ldots,v_t\}$, say without loss of generality, $v_1,\ldots v_{21}$, such that 
$v_i$, for $i=1,\ldots,21$, is adjacent to a distinct vertex $x_i\neq u$.

For $i=1,\ldots,21$, let $A_i$ be a bag of the path decomposition of $H$ that contains $v_i$ and $u$ (such a bag exists as $uv_i\in E_H$). Then, for $2\leq i \leq 20$, we may assume that $A_i=\{v_i,x_i,u\}$. Note that $x_i$ and $x_{i+1}$ might be adjacent in $H$, but $x_i$ cannot be a neighbour of $x_k$, with $k\geq i+2$, because of the existence of bag $\{v_{i+1},x_{i+1},u\}$. 
For the same reason, $x_i$ cannot be adjacent to $v_k$ for some $k\geq i+2$. Also, if $k\geq i+2$, all paths in $H$ from $x_i$ to $x_k$ contain either $x_{i+1}$ or $v_{i+1}$. The same applies for the paths from $x_i$  to $v_k$ for some $k\geq i+2$.  
Then $\{x_1,x_6,x_{11},x_{16},x_{21}\}$ are vertices that are pairwise at distance at least~3 in $G-u$.  
\end{proof}

\begin{lemma} \label{lem:pwlabelu}
Let $G$ be a $c_1$-twin-bounded graph that has a minimal pathwidth-$2$ root~$H$. 
Let $u$ be a vertex such that there are five distinct vertices $x_1,\ldots,x_5\in N_G(u)$ that are pairwise at distance at least~$3$
in $G-u$. Then, for any $x\in N_G(u)$, it holds that $xu\notin E_H$ if and only if $\dist_{G-u}(x,x_i)\geq 3$ for some $1\leq i\leq 5$.
\end{lemma}

\begin{proof}
Let $x\in N_G(u)$. First suppose that $\dist_{G-u}(x,x_i)\geq 3$ for some $1\leq i\leq 5$. 
Then, by Lemma~\ref{lem:not-incl} we find that $xu\notin E_H$. 

Now suppose that $xu\notin E_H$. If $x=x_j$ for some $j\in \{1,\ldots,5\}$, then $\dist_{G-u}(x,x_i)\geq 3$ for $i\neq j$.
Hence we may assume that
$x\notin \{x_1,\ldots,x_5\}$. As $\dist_{G-u}(x_i,x_j)\geq 3$ for $1\leq i<j\leq 5$, Lemma~\ref{lem:not-incl} tells us that  $ux_i\notin E_H$ 
for $i=1,\ldots,5$. As $ux_i\in E_G$ for $i=1,\ldots,5$, this means that for $i=1,\ldots,5$, there exists a vertex~$v_i$ such that $x_iv_i\in E_H$ and $v_iu\in E_H$. 
We observe that
$v_iv_j\notin E_H$ and $v_ix_j\notin E_H$ for $i\neq j$, 
as otherwise $\dist_{G-u}(x_i,x_j)\leq 2$. 
Assume that $v_1,\ldots,v_5$ appear in this order in the path decomposition of~$H$. Then the path decomposition of $H$ contains
the sets $\{v_2,x_2,u\}$, $\{v_3,x_3,u\}$ and $\{v_4,x_4,u\}$ as bags. 

First consider the case where $x$ appears before $v_2$ in the path decomposition of~$H$. 
If a shortest path between $x$ and $x_4$ in $G-u$ contains
$y\in \{x_2,x_3\}$, then $\dist_{G-u}(x,x_4)\geq \dist_{G-u}(y,x_4)+1\geq 4$.
Otherwise a shortest path between $x$ and~$x_4$ in $G-u$ must contain either $v_2$, $v_3$, which are both not adjacent to $x_4$ in $G-u$, or another neighbour of $u$ that appeared previously in the path decomposition and has no common neighbour with $x_4$ in $H$.
Assume without loss of generality that it contains $v_2$.
We have $\dist_{G-u}(v_2,x_4)\geq 2$, as otherwise $\dist_{G-u}(x_2,x_4)<3$. As $\dist_{G-u}(v_2,x_4)\geq 2$, we obtain $\dist_{G-u}(x,x_4)\geq 3$. 

Now consider the case where $x$ appears between $\{v_2,x_2,u\}$ and $\{v_3,x_3,u\}$. Then we consider $x_5$ instead of $x_4$. By the same argument as above we find that $\dist_{G-u}(x,x_5)\geq 3$ due to the existence of bags $\{v_3,x_3,u\}$ and $\{v_4,x_4,u\}$. The other cases follow by symmetry.      
\end{proof}

Let $G$ be a $c_1$-twin-bounded graph that has a minimal pathwidth-$2$ root~$H$.
We define the following two sets for a vertex $u$ with $d_H(u)\geq c_2$:
$$R_u=\{ w\in N_G(u)~|~uw\in E_H\} \mbox{ and } B_u=\{ w\in N_G(u)~|~uw\notin E_H\}.$$
For a vertex $v\in B_u$ we define the set 
$$X_v=\{ x\in R_u ~|~ vx\in E_G\}.$$
Using the above notions we prove the following lemma, in which we identify edges that do not belong to a minimal pathwidth-2 root.

\begin{lemma} \label{lem:pwexclude}
Let $G$ be a $c_1$-twin-bounded graph that has a minimal pathwidth-$2$ root~$H$.
Let $x,y\in N_H(u)$ for some vertex $u$ with $d_H(u)\geq c_2$. If there is no vertex $v\in B_u$ with $x,y\in X_v$, then $xy\notin E_H$.
\end{lemma}

\begin{proof}
We prove the lemma by contraposition. Assume that $xy\in E_H$.
By Lemma~\ref{obs:triangle}, there exists a vertex $v$ that is, in $H$, adjacent to at least one of $x,y$, but not to $u$. The latter implies that $v\in B_u$. Hence the set $X_v$ is defined. Say $vx\in E_H$, which implies that $vx\in E_G$. As $xy\in E_H$, we also find that $vy\in E_G$. Hence, as $\{x,y\}\subseteq R_u$, both $x$ and $y$ are in $X_v$.
\end{proof}

We will also need the following lemma.

\begin{lemma}\label{lem:c3}
Let $G$ be a $c_1$-twin-bounded graph that has a minimal pathwidth-$2$ root~$H$.
Let $v\in B_u$ for some vertex $u$ with $d_H(u)\geq c_2$. 
Then the number of bags in the path decomposition of $H$ containing $u$ and a vertex of $X_v$ is at most $c_3=15c_1+4$.
\end{lemma}

\begin{proof}
Let $A$ and $A'$ be the first and the last bag in the path decomposition of $H$ containing $u$ and a vertex of $X_v$. 
First suppose that $v$ belongs to both $A$ and $A'$. Then all the bags between $A$ and $A'$ (including $A$ and $A'$ themselves) contain both $u$ and $v$. 
Recall that $X\nsubseteq Y$ for any two bags $X$ and $Y$. Hence, for every vertex appearing between $A$ and $A'$ we have exactly one new bag. 
By Lemma~\ref{lem:pwboundxi}, the number of such vertices, and thus the number of bags between $A$ and $A'$, is at most $3c_1$.

Now suppose that $v$ appears before $A$ but is not contained in $A$. By definition, $A'$ contains a vertex $x\in X_v$. Since the bags containing $v$ appear before $A$, we find that $xv\notin E_H$. As $x\in X_v$, this means that $\dist_H(x,v)=2$.
Hence there exists a vertex $y$ such that $xy,yv\in E_H$. This means that there exists a bag containing $\{v,y\}$. As this bag contains~$v$, it is before $A$ in the path decomposition of $H$.
It also means that there exists a bag 
$\{x,y,u\}$.
As $x\in X_v$, this bag must be between $A$ and $A'$. 
As bags between $A$ and $\{x,y,u\}$  contain $\{y,u\}$, there are at most  $3c_1$ of them  due Lemma~\ref{lem:pwboundxi}.
By the same arguments as in the first case, the number of bags between $\{x,y,u\}$ and $A'$ is at most $3c_1$ as well.
Hence, the number of bags between $A$ and $A'$ is at most $6c_1+1$.
By symmetry, we find the same bound if $v$ appears after $A'$ but is not contained in $A'$.

Now suppose that $v$ belongs to $A$ but not to $A'$.
Let $x\in X_v$ be such that $x\in A'$. 
If $xv\in E_H$, the number of bags between $A$ and $A'$ can again be bounded by $6c_1+1$, by a similar argument as used in the previous case. Assume $xv\notin E_H$ and
let $xyv$ be a path between $x$ and $v$. There exists a bag containing $\{v,y,u\}$ and a bag containing $\{y,x,u\}$ that appears after $\{v,y,u\}$. 
By the same arguments as before, the constant $3c_1$ bounds the number of bags between $A$ and $\{v,y,u\}$; between $\{v,y,u\}$ and $\{y,x,u\}$; and between $\{y,x,u\}$ and $A'$.
Hence, the number of bags between $A$ and $A'$ is at most $9c_1+2$.
By symmetry, we find the same bound if $v$ belongs to $A'$ but not to $A$.

Finally suppose that $v$ appears between $A$ and $A'$ but is not contained in them. We proceed in the same way as before,
and the worst scenario is when the vertices of $X_v$ contained in $A$ and $A'$ are not adjacent to $v$. Let $x\in X_v\cap A$ and $y\in V_H$ be such that $xyv$ is a path between $x$ and $v$. We take $x'$ and $y'$ analogously with respect to $A'$. By Lemma~\ref{lem:pwboundxi}, the constant $3c_1$ bounds the number of bags between the following pairs of bags: $A$ and $\{x,y,u\}$; $\{x,y,u\}$ and $\{v,y,u\}$; $\{v,y,u\}$ and $\{v,y',u\}$; $\{v,y',u\}$ and $\{x',y',u\}$; and $\{x',y',u\}$ and $A'$.  
The total number of bags between $A$ and $A'$ is therefore at most $c_3=15c_1+4$. 
\end{proof}

Let $G$ be a $c_1$-twin-bounded graph that has a minimal pathwidth-$2$ root~$H$. 
Let $U$ be the set of vertices of $H$ with $d_H(u)\geq c_2$. 
For every $u\in U$ we do the following:
\begin{itemize}
\item for every two distinct vertices $x,y\in N_H(u)$ for which no vertex $v\in B_u$ exists with $x,y\in X_v$, delete the edge $xy$ from $G$ (note that this edge exists in $G$).
\end{itemize}
We denote the resulting graph by $G_H$; note that $G_H$ is a spanning subgraph of $G$.
We now prove, in our last structural lemma, that the class of graphs~$G_H$ has bounded pathwidth.

\begin{lemma} \label{lem:pwbdtw}
Let $G$ be a $c_1$-twin-bounded graph that has a minimal pathwidth-$2$ root~$H$.
Let $U$ be the set of vertices of $H$ with $d_H(u)\geq c_2$. 
Then  $\pw(G_H)\leq c_4$ for $c_4=3(c_2-1)^{\lfloor\frac{c_3+1}{2}\rfloor+1}$.
\end{lemma} 

\begin{proof}
For each $u\in U$, we do the following. Consider the bags $B_1,\ldots, B_t$ in the path decomposition of $H$ containing $u$ and its neighbours. Starting from $B_1$, we pick the first bag where a new neighbour of $u$ appears. Let $B_i$ be such a bag. As $B_i$ contains at least one vertex that is not contained in $B_{i-1}$, we have $|B_i\cap B_{i-1}|\leq 2$, while we already know that $u\in B_i\cap B_{i-1}$. 
In the bags $B_1,\ldots B_{i-1}$, we replace~$u$ by a new vertex~$u_1$. We create a new bag between $B_{i-1}$ and $B_i$ containing $u_1$, $u_2$ and $(B_i\cap B_{i-1})\setminus \{u\}$. In the bags $B_i,\ldots, B_t$, we replace $u$ by $u_2$. In general, for every bag $B_k$ found containing a new neighbour of $u$ we do the following:
\begin{enumerate}
\item Create a new bag between $B_{k-1}$ and $B_k$ containing $u_{j+1}$ and the vertices of $B_{k-1}\cap B_k$ (note that $u_j\in B_{k-1}\cap B_k$).
\item In the bags $B_k,\ldots, B_t$, replace $u_j$ by $u_{j+1}$.
\item In $H$, add an edge between $u_j$ and $u_{j+1}$ and an edge between $u_{j+1}$ and the newly found neighbour of $u$. 
\end{enumerate} 
Let $\hat{H}$ be the graph obtained from $H$ by the above procedure. Note that $H$ is a contraction of $\hat{H}$, as $H$ can be obtained by contracting the edges of the paths created for each vertex of $U$. As we constructed a path decomposition of $\hat{H}$ with the same width as the one  $H$, we have $\pw(\hat{H})\leq 2$. 

If $v\in V_H\setminus U$, then $d_H(v)<c_2$
and, in each step of the above procedure, the degree of $v$ is maintained. The vertices $u_i$ created for each vertex of $U$ have degree 
at most~$3<c_2$.
Thus the graph $\hat{H}$ has degree
at most $c_2-1$.

We claim that $G_H$ is a minor of $\hat{H}^{c_3+1}$. Let $\hat{G}$ be obtained from $\hat{H}^{c_3+1}$ by 
contracting all edges 
of the paths created for each vertex of $U$. We may assume that $V_{\hat{G}}=V_{G_H}$ and show below that $G_H$ is a subgraph of $\hat{G}$. 

Every edge of $G_H$ that belongs to $E_H$ is also an edge of $\hat{H}^{c_3+1}$. Let $xy\in E_{G_H}$ be such that $xy\notin E_H$. As $H$ is a square root of $G$, there exists $u\in V_{G_H}$ such that $xu,yu\in E_H$. Let $X'$ and $Y'$ be the sets of vertices of $\hat{H}$ that were contracted to $x$ and $y$, respectively. If $u\notin U$, then by the construction of $\hat{H}$ there are vertices $x'\in X'$ and $y'\in Y'$ such that $x'u,y'u\in E_{\hat{H}}$ and therefore $x'y'\in\hat{H}^{c_3+1}$ and $xy\in E_{\hat{G}}$. If $u\in U$, there exists a path $u_i\ldots u_j$ in $\hat{H}$ and vertices $x'\in X$ and $y'\in Y$ such that $x'u_i\in E_{\hat{H}}$ and $u_jy'\in E_{\hat{H}}$. Since $xy\in E_{G_H}$ and $u\in U$, we know that $x,y\in X_v$ for some $v$, otherwise we would have deleted the edge $xy$ when constructing $G_H$. As the number of bags containing $u$ and vertices of $X_v$ is at most~$c_3$ by Lemma~\ref{lem:c3}, the length of the path $u_i\ldots u_j$ is at most~$c_3$. This implies that $\dist_{\hat{H}}(x',y')\leq c_3+1$ 
and hence $x'y'\in E_{\hat{H}^{c_3+1}}$, which in turn implies that $xy \in E_{\hat{G}}$.
Since $G_H$ is a subgraph of $\hat{G}$ and $\hat{G}$ is a contraction of $\hat{H}^{c_3+1}$, we conclude that $G_H$ is a minor of $\hat{H}^{c_3+1}$.

As $\pw(\hat{H})\leq 2$
and $\hat{H}$ has bounded degree, we find that $\pw(\hat{H}^{c_3+1})\leq 3(c_2-1)^{\lfloor\frac{c_3+1}{2}\rfloor+1}$  due to Lemma~\ref{lem:tw-power}.
 Since $G_H$ is a minor of $\hat{H}^{c_3+1}$, we find that $\pw(G_H)\leq \pw(\hat{H}^{c_3+1})$ due to Lemma~\ref{obs:minor}.
Hence, $\pw(G_H)  \leq 3(c_2-1)^{\lfloor\frac{c_3+1}{2}\rfloor+1}$ and we can take $c_4=3(c_2-1)^{\lfloor\frac{c_3+1}{2}\rfloor+1}$.    
\end{proof}

\subsection{The Algorithm}\label{secpw:algo}

In this section, we construct our $O(n^6)$-time algorithm for \textsc{Pathwidth-$2$ Root}, that is, we are now ready to prove Theorem~\ref{thm:pw2}. 
In order to dot this we follow the proof of Theorem~\ref{thm:outerplanar} and replace in that proof the basic results for outerplanar graphs from
Section~\ref{s-square} and the structural results for graphs with outerplanar roots from Section~\ref{sec:tech} with the basic results for graphs of pathwidth at most~2 from Section~\ref{s-treepath} and the structural results for graphs with pathwidth-2 roots from Section~\ref{secpw:strc}.

\medskip
\noindent
{\bf  Theorem~\ref{thm:pw2} (restated).}
{\it \textsc{Pathwidth-$2$ Root} can be solved in $O(n^6)$ time.}

\begin{proof}
Let $G$ be the input graph. We may assume without loss of generality that $G$ is connected and has $n\geq 2$ vertices. We first exhaustively apply the following rule in order to reduce the number of true twins each vertex can have in a (potential) 
pathwidth-2 root of $G$.

\medskip
\noindent
{\bf Deleting a true twin.} If $G$ has a set $X$ of true twins of size at least $c_1+1$,
then delete an arbitrary vertex $u\in X$ from $G$.

\medskip
\noindent
The following claim shows that this rule is safe.

\medskip
\noindent
{\bf Claim 1.}
{\it If $G'=G-u$ is obtained from $G$ by the application of {\bf deleting a true twin}, then $G$ has a pathwidth-$2$ root if and only if $G'$ has a pathwidth-$2$ root.}

\medskip
\noindent
We proof Claim~1 as follows.
First suppose that $G$ has a pathwidth-$2$ root~$H$. 
We may assume without loss of generality that $H$ is minimal.
Note that $H-z$ has pathwidth at most~2 for every $z\in V_H$.
Since $|W|\geq c_1+1$, there is a vertex $v\in W$ 
satisfying condition~(i) of Lemma~\ref{lem:pwreduction1} or there are three vertices $v_1,v_2,v_3\in W$ satisfying condition~(ii) of Lemma~\ref{lem:pwreduction1}.
As the vertices of $W$ are true twins, we take $u=v$ in the first case and $u=v_1$ in the second case to find that
 $H-u$ is a pathwidth-2 root of $G-u$.

Now suppose that $G-u$ has a pathwidth-2 root~$H'$, which we may assume to be minimal.
Since $|W\setminus \{u\}|\geq c_1$, there is a vertex $v$ satisfying condition~(i) of Lemma~\ref{lem:pwreduction1}
or there are three vertices $v_1,v_2,v_3\in W$ satisfying condition~(ii) of Lemma~\ref{lem:pwreduction1}. 

In the first case, let $w$ be the (unique) vertex of $H'$ that is adjacent to~$v$. We add $u$ and the edge $uw$ to $H'$ to obtain a square root $H$ of $G$. We still need to prove that $\pw(H)\leq 2$. We may assume that $v$ appears in only one bag (which also contains $w$) in the path decomposition of $H'$. Otherwise we can delete all other occurrences of $v$ and obtain another path decomposition of $H'$ that has width at most~2. Let $A_i$ be the bag containing $\{v,w\}$, and let $A_{i+1}$ be the next bag of the path decomposition. If $w\in A_{i+1}$, then  we create a new bag between $A_i$ and $A_{i+1}$ containing $(A_i\cap A_{i+1})\cup \{u\}$. If $w\notin A_{i+1}$, then $|A_i\cap A_{i+1}|\leq 1$, and the new bag will contain $(A_i\cap A_{i+1})\cup \{u,w\}$. Note that in both cases
the new bag contains at most three vertices. Hence we obtained a path decomposition of $H$ that has width at most~2. 

In the second case, let $N_H(v_1)=N_H(v_2)=N_H(v_3)=\{w,y\}$.
We add $u$ and the edge $uw$, $uy$ to $H'$ to obtain a square root $H$ of $G$. We still need to prove that $\pw(H)\leq 2$.
Since $N_H(v_1)=N_H(v_2)=N_H(v_3)=\{w,y\}$, the path decomposition of $H'$ contains a bag~$A_i=\{w,y,v_i\}$ for some $i\in\{1,2,3\}$. Since $v_i$ is only adjacent to $w$ and $y$, we may assume that $A_i$ is the only bag in the path decomposition containing $v_i$. Let $A_{i+1}$ be the next bag of the path decomposition. We create a new bag $\{u,w,y\}$ between $A_i$ and $A_{i+1}$ to obtain a path decomposition of $H$ that has width at most~2.      
This proves Claim~1.

\medskip
\noindent
For simplicity, we call the graph obtained by exhaustive application of {\bf deleting a true twin} $G$ again.
The next claim immediately follows from the rule {\bf deleting a true twin}.

\medskip
\noindent
{\bf  Claim 2.} {\it The graph $G$ is $c_1$-twin-bounded.}

\medskip
\noindent
In the next stage of our algorithm we are going to label some edges of $G$ \emph{red} or \emph{blue} in such a way that the red edges are included in every minimal pathwidth-2 root of $G$, whereas the blue edges are excluded from any minimal pathwidth-2 root
 of~$G$. We let $R$ denote the set of red edges and $B$ the set of blue edges. We also construct a set of vertices~$U$ of~$G$ such that for every $u\in U$, the edges incident to $u$ are labeled red or blue. 

\medskip
\noindent
{\bf Labeling edges.} Set $U=\emptyset$, $R=\emptyset$ and $B=\emptyset$. For each $u\in V_G$ such that there are five distinct vertices $v_1,\ldots,v_5 \in N_G(u)$ that are at distance at least~$3$ from each other in $G-u$, do the following:
\begin{itemize}
\item[(i)] set $U=U\cup\{u\}$;
\item[(ii)] set $B'=\{ux\in E_G\mid \text{there is an } 1\leq i\leq 5\text{ such that }\dist_{G-u}(x,v_i)\geq 3\}$;
\item[(iii)] set $R'=\{ux\mid x\in N_G(u)\}\setminus B'$;
\item[(iv)] set $R=R\cup R'$ and $B=B\cup B'$;
\item[(v)] if $R\cap B\neq\emptyset$, then return a no-answer and stop.
\end{itemize}
Note that the above rule does not change the graph $G$ itself. 
Lemmas~\ref{lem:pwfive} and \ref{lem:pwlabelu}, combined with Claim~2, imply the following claim.

\medskip
\noindent
{\bf Claim 3.}
{\it If $G$ has a minimal pathwidth-$2$ root~$H$, then {\bf labeling edegs} does not stop in step~(v). Moreover, $R\subseteq E_H$ and $B\cap E_H=\emptyset$, and every vertex $u\in V_G$ with $d_H(u)\geq c_2$ is included in $U$.}

\medskip
\noindent
Next, we are going to find, for each $u\in U$, a set~$S$ of edges $xy$ with $xu,yu\in R$ that may be removed from $G$.

\medskip
\noindent
{\bf Deleting irrelevant edges.} Set $S=\emptyset$. For each $u\in U$ and every pair of distinct vertices $x,y\in N_G(u)$ such that $xu,uy\in R$ 
do the following:
\begin{itemize}
\item[(i)] if $xy\notin E_G$, then return a no-answer and stop;
\item[(ii)] if there is no $v\in N_G(u)$ such that $vu\in B$ and $x,y\in  N_G(v)$, then include $xy$ in $S$;
\item[(iii)] if $R\cap S\neq\emptyset$, then return a no-answer and stop;
\item[(iv)] remove the edges of $S$ from $G$.
\end{itemize}

By  combining Lemma~\ref{lem:pwexclude} with Claim~3 we obtain the following claim.

\medskip
\noindent
{\bf Claim 4.}
{\it If $G$ has a minimal pathwidth-$2$ root~$H$, then {\bf deleting irrelevant edges} does not stop in step~(i) or~(iii), and moreover,  $S\cap E_H=\emptyset$.}

Assume that we have not stopped and returned a no-answer after the execution of {\bf deleting irrelevant edges}.
Let $G'=G-S$. Again we find that a square root of $G$ may not be a square root of $G'$ and vice versa. However, we can prove the following claim.

\medskip
\noindent
{\bf Claim~5.}
{\it The graph~$G$
has a pathwidth-$2$ root if and only if there is a set $L\subseteq E_{G'}$ such that 
\begin{itemize}
\item[(i)] $R\subseteq L$ and $B\cap L=\emptyset$;
\item[(ii)] for every $xy\in E_{G'}$,  $xy\in L$ or there exists a vertex $z\in V_{G'}$ with $xz,zy\in L$;
\item[(iii)] for every two distinct edges $xz,yz\in L$, it holds that $xy\in E_{G'}$ or there is a vertex $u\in U$ with $ux,uy\in R$;
\item[(iv)] the graph $H=(V_G,L)$ has pathwidth at most~$2$.
\end{itemize}}

\medskip
\noindent
We prove Claim~5 as follows.
First suppose that $H$ is a minimal outerplanar root of $G$. By Claim~4 we find that
$E_H\cap S=\emptyset$, that is, $E_H\subseteq E_{G'}$. 
Let $L=E_H$.
Then (i) holds due to Claim~3, whereas (ii) and (iv) hold because $H=(V_G,L)$ is a pathwidth-2 root of $G$. 
To prove (iii) suppose that $xz$ and $zy$ are distinct edges of $L$ such that $xy\notin E_{G'}$.
As $H=(V_G,L)$ is a square root of $G$, this means that $xy\in E_G\setminus E_{G'}$, that is, $xy\in S$.
By definition of the rule {\bf deleting irrelevant edges}, this means that there must exist a vertex $u\in U$ such that $xu,uy\in R$.

Now suppose that there is a subset $L\subseteq E_{G'}$ such that (i)--(iv) hold.
Let $xy\in E_G$. If $xy\in E_{G'}$, then  $xy\in L$ or there is a vertex $z\in V_{G'}$ such that $xz,yz\in L$ by~(ii).
If $xy\in E_G\setminus E_{G'}=S$, then there is a vertex $u\in U$ such that $xu,uy\in R$ by (iii).
As $R\subseteq L$ by~(i), 
we find that $xu,uy\in L$. Hence $G$ is a subgraph of $(V_G,L)^2$. As $L\subseteq E_{G'}$, we find that $G=(V_G,L)^2$.
We conclude that $H=(V_G,L)$ is a square root of $G$. By (iv) we find that $H$ is a pathwidth-2 root of $G$.
Hence we have proven Claim~5.

\medskip
\noindent
It remains to check the existence of a set of edges $L$ satisfying (i)--(iv) of Claim~5 for a given triple $G'$, $R$, $B$,  
which is the final step of the algorithm. Notice that 
If $G$ has a minimal pathwidth-2 root $H$, then $G'$ is a subgraph of $G_H$ constructed in Section~\ref{secpw:strc}; this is due to
 Lemmas~\ref{lem:pwfive} and \ref{lem:pwexclude}. 
 By Lemma~\ref{lem:pwbdtw}, we find that $\pw(G')\leq \pw(G_H)\leq c_4$. 
 Hence we must return a no-answer and stop if $\pw(G') > c_4$.
 
Now suppose  $\pw(G')\leq c_4$. As $\tw(G')\leq \pw(G')$, this means that $\tw(G')\leq c_4$.
It is straightforward to verify that properties (i)--(iv) in Claim~5 can be expressed in MSO. In particular, to express outerplanarity in (iv), we combine  
Lemma~\ref{lem:pw2c} with Lemma~\ref{l-engel}. Afterwards we use Lemma~\ref{l-courcelle}.

The correctness of our algorithm follows from the above description and proofs of Claims~1--5.
It remains to evaluate the running time of our algorithm, which we do below.

\medskip
\noindent
We can verify in $O(n)$ time if two vertices of $G$ are true twins. This means that the classes of true twins can be constructed in $O(n^3)$ time. Therefore, the exhaustive application of {\bf deleting a simplicial true twin} costs $O(n^3)$ time.
For every vertex $u$, we can compute the distances between the vertices of $N_G(u)$ in $G-u$ in $O(n^3)$ time. 
This implies that {\bf labeling edges} can be done in $O(n^6)$ time. 
Applying {\bf deleting irrelevant edges} takes $O(n^4)$ time, as it takes $O(n^2)$ to process a pair $x,y$ and the number of 
such pairs is $O(n^2)$. We construct $G'$  in linear time. Finally, checking whether $\tw(G')\leq 3\cdot c_4$ and deciding whether there is a set of edges $L$ satisfying the required properties can be done in linear time by Lemma~\ref{l-bod} and~\ref{l-courcelle}, respectively. 
Hence the total running time is $O(n^6)$. This completes the proof of Theorem \ref{thm:pw2}.
\end{proof}

Similarly to {\sc Outerplanar Root}, we remark that one can find a a pathwidth-2 root of a graph if it exists using a dynamic programming algorithm.

\section{Conclusions}\label{s-con}
We proved that {\sc ${\cal H}$-Square Root} is polynomial-time solvable when ${\cal H}$ is the class of outerplanar graphs or the class of graphs of pathwidth at most~2. 
In fact, our technique allows us to obtain results that are more general than Theorems~\ref{thm:outerplanar} and~\ref{thm:pw2}.
Namely, we can  solve $\mathcal{H}$-{\sc Square Root} in polynomial time for every subclass~${\cal H}$ of outerplanar graphs or graphs of pathwidth at most~2, respectively, that satisfies the following two conditions:
\begin{itemize}
\item [(i)] ${\cal H}$ is closed 
under vertex deletion and edge deletion, and
\item [(ii)] ${\cal H}$ can be defined in CMSO.
\end{itemize}

We briefly sketch how this generalization 
can be obtained for subclasses of outerplanar graphs that satisfy conditions~(i) and~(ii). The proof for subclasses of pathwidth at most~2 is similar.

Let ${\cal H}$ be a subclass of outerplanar graphs that satisfy conditions~(i) and~(ii).
It is straightforward to show  
the result
if $\mathcal{H}$ is closed under {\it pendant vertex addition}, which means that every graph obtained from a graph $H\in\mathcal{H}$ by creating a new vertex and making it adjacent to a vertex of $H$ belongs to $\mathcal{H}$. In this case, we can simply repeat the proof of Theorem~\ref{thm:outerplanar}, as this property, together with condition~(i) 
ensures that {\bf deleting a simplicial true twin} is safe, while condition~(ii) guarantees that the remaining part of the algorithm remains correct. 

However, if $\mathcal{H}$ is not closed under pendant vertex addition, then we cannot claim that {\bf deleting a simplicial true twin} is sound. We can still show that the graph
$G-u$, where $u$ is a twin vertex, has a square root $H'\in {\cal H}$ if $G$ has a square root $H\in {\cal H}$, but the opposite might be false.
The reason is that  we cannot duplicate a pendant vertex of $H'$ to obtain a square root of $G$. This situation happens, for example, if $\mathcal{H}$ is a class of outerplanar graphs of bounded degree. 
To overcome this difficulty, we need some additional properties of CMSO. In particular, it is known that every CMSO-definable property on structures has a finite state. This fact was first explicitly proved by Bodlaender et al. in~\cite{BodlaenderFLPST16} and we refer to this paper for the definitions. Lemma~3.2 of~\cite{BodlaenderFLPST16} implies the following lemma.

\begin{lemma}\label{lem:finite-state}
Let $\varphi$ be a CMSO formula on graphs. For every positive integer~$d$, there exists positive integers $s$ and $t$ with $s<t$ that only depend on $\varphi$ and $d$, such that the following holds: if a graph $H$ has a family $X$ of false twins of degree~$d$, such that $|X|\geq t$ and $Y\subset X$ with $|Y|=s$, then $H\models\varphi$ if and only if $H-Y\models \varphi$.
\end{lemma}

We use Lemma~\ref{lem:finite-state} to modify the  {\bf deleting a simplicial true twin} rule as follows. Let $\varphi$ be a CMOS formula 
 such that
$H\in\mathcal{H}$ if and only of $H\models\varphi$. We take the constants $s$ and $t$ for $\varphi$ and $d=1$. Then we construct the new rule:

\medskip
\noindent
{\bf Deleting a simplicial true twin$^*$.} 
If $G$ has a set $X$ of simplicial true twins of size at least~$t+7$, then delete the vertices of an arbitrary set $Y\subset X$ of size~$s$ from $G$.

\smallskip
\noindent
By using the same arguments as in the proof of Claim~1, we can show that if $G'=G-Y$ is obtained from $G$ by the application of {\bf deleting a simplicial true twin$^*$}, then $G$ has a square root $H\in \mathcal{H}$ if and only if $G'$ has a square root $H'\in\mathcal{H}$.
Afterwards we apply the same {\bf labeling edges} and {\bf deleting irrelevant edges} rules and show Claim~5 in the same way as before 
(namely, by using the fact that condition~(i) holds).
For the final stage, we have to adjust the constant upper bound on the treewidth, which has increased due the modified rule of deleting simplicial true twins.

\medskip
\noindent
We conclude our paper by posing the following two open problems. First, is {\sc ${\cal H}$-Square Root} polynomial-time solvable for every class $\mathcal{H}$ of graphs
of bounded pathwidth?
Second, is {\sc ${\cal H}$-Square Root} polynomial-time solvable if ${\cal H}$ is the class of planar graphs?
Both these problems require additional proof techniques to solve them.

\medskip
\noindent
{\it Acknowledgements.} We thank Dimitrios M. Thilikos for helpful comments 
on the generalizations of Theorems~\ref{thm:outerplanar} and~\ref{thm:pw2} in Section~\ref{s-con},
and we thank an anonymous reviewer for helpful comments on our paper.

\end{document}

%% file: Fig1.pdf_t
\begin{picture}(0,0)%
\includegraphics{Fig1.pdf}%
\end{picture}%
\setlength{\unitlength}{3947sp}%
\begingroup\makeatletter\ifx\SetFigFont\undefined%
\gdef\SetFigFont#1#2#3#4#5{%
  \reset@font\fontsize{#1}{#2pt}%
  \fontfamily{#3}\fontseries{#4}\fontshape{#5}%
  \selectfont}%
\fi\endgroup%
\begin{picture}(5774,3416)(523,-2208)
\put(5431,794){\makebox(0,0)[lb]{\smash{{\SetFigFont{12}{14.4}{\rmdefault}{\mddefault}{\updefault}{\color[rgb]{0,0,0}$X$}%
}}}}
\put(1704,-2128){\makebox(0,0)[lb]{\smash{{\SetFigFont{12}{14.4}{\rmdefault}{\mddefault}{\updefault}{\color[rgb]{0,0,0}$u=v_1$}%
}}}}
\put(4611,-2139){\makebox(0,0)[lb]{\smash{{\SetFigFont{12}{14.4}{\rmdefault}{\mddefault}{\updefault}{\color[rgb]{0,0,0}$u$}%
}}}}
\put(1151,-2008){\makebox(0,0)[lb]{\smash{{\SetFigFont{12}{14.4}{\rmdefault}{\mddefault}{\updefault}{\color[rgb]{0,0,0}$v_2$}%
}}}}
\put(538,-1515){\makebox(0,0)[lb]{\smash{{\SetFigFont{12}{14.4}{\rmdefault}{\mddefault}{\updefault}{\color[rgb]{0,0,0}$v_3$}%
}}}}
\put(2451,-1922){\makebox(0,0)[lb]{\smash{{\SetFigFont{12}{14.4}{\rmdefault}{\mddefault}{\updefault}{\color[rgb]{0,0,0}$v_n$}%
}}}}
\put(3871,415){\makebox(0,0)[lb]{\smash{{\SetFigFont{12}{14.4}{\rmdefault}{\mddefault}{\updefault}{\color[rgb]{0,0,0}$x_1$}%
}}}}
\put(4624,675){\makebox(0,0)[lb]{\smash{{\SetFigFont{12}{14.4}{\rmdefault}{\mddefault}{\updefault}{\color[rgb]{0,0,0}$x_2$}%
}}}}
\put(5891,-145){\makebox(0,0)[lb]{\smash{{\SetFigFont{12}{14.4}{\rmdefault}{\mddefault}{\updefault}{\color[rgb]{0,0,0}$x_3$}%
}}}}
\put(5784,-1398){\makebox(0,0)[lb]{\smash{{\SetFigFont{12}{14.4}{\rmdefault}{\mddefault}{\updefault}{\color[rgb]{0,0,0}$x_k$}%
}}}}
\end{picture}%

%% file: Fig2.pdf_t
\begin{picture}(0,0)%
\includegraphics{Fig2.pdf}%
\end{picture}%
\setlength{\unitlength}{3947sp}%
\begingroup\makeatletter\ifx\SetFigFont\undefined%
\gdef\SetFigFont#1#2#3#4#5{%
  \reset@font\fontsize{#1}{#2pt}%
  \fontfamily{#3}\fontseries{#4}\fontshape{#5}%
  \selectfont}%
\fi\endgroup%
\begin{picture}(2336,2188)(1229,-2252)
\put(3550,-606){\makebox(0,0)[lb]{\smash{{\SetFigFont{12}{14.4}{\rmdefault}{\mddefault}{\updefault}{\color[rgb]{0,0,0}$x_3$}%
}}}}
\put(2610,-1440){\makebox(0,0)[lb]{\smash{{\SetFigFont{12}{14.4}{\rmdefault}{\mddefault}{\updefault}{\color[rgb]{0,0,0}$u$}%
}}}}
\put(2077,-247){\makebox(0,0)[lb]{\smash{{\SetFigFont{12}{14.4}{\rmdefault}{\mddefault}{\updefault}{\color[rgb]{0,0,0}$x_1$}%
}}}}
\put(2677,-273){\makebox(0,0)[lb]{\smash{{\SetFigFont{12}{14.4}{\rmdefault}{\mddefault}{\updefault}{\color[rgb]{0,0,0}$x_2$}%
}}}}
\end{picture}%

%% file: Fig3.pdf_t
\begin{picture}(0,0)%
\includegraphics{Fig3.pdf}%
\end{picture}%
\setlength{\unitlength}{3947sp}%
\begingroup\makeatletter\ifx\SetFigFont\undefined%
\gdef\SetFigFont#1#2#3#4#5{%
  \reset@font\fontsize{#1}{#2pt}%
  \fontfamily{#3}\fontseries{#4}\fontshape{#5}%
  \selectfont}%
\fi\endgroup%
\begin{picture}(3973,2625)(836,-2503)
\put(851,-1328){\makebox(0,0)[lb]{\smash{{\SetFigFont{12}{14.4}{\rmdefault}{\mddefault}{\updefault}{\color[rgb]{0,0,0}$v_1=v_1'$}%
}}}}
\put(3117,-2434){\makebox(0,0)[lb]{\smash{{\SetFigFont{12}{14.4}{\rmdefault}{\mddefault}{\updefault}{\color[rgb]{0,0,0}$u$}%
}}}}
\put(4784,-1795){\makebox(0,0)[lb]{\smash{{\SetFigFont{12}{14.4}{\rmdefault}{\mddefault}{\updefault}{\color[rgb]{0,0,0}$v_3=v_3'$}%
}}}}
\put(3270,-61){\makebox(0,0)[lb]{\smash{{\SetFigFont{12}{14.4}{\rmdefault}{\mddefault}{\updefault}{\color[rgb]{0,0,0}$v_2$}%
}}}}
\put(3263,-581){\makebox(0,0)[lb]{\smash{{\SetFigFont{12}{14.4}{\rmdefault}{\mddefault}{\updefault}{\color[rgb]{0,0,0}$v_2'$}%
}}}}
\end{picture}%

%% file: Fig4.pdf_t
\begin{picture}(0,0)%
\includegraphics{Fig4.pdf}%
\end{picture}%
\setlength{\unitlength}{3947sp}%
\begingroup\makeatletter\ifx\SetFigFont\undefined%
\gdef\SetFigFont#1#2#3#4#5{%
  \reset@font\fontsize{#1}{#2pt}%
  \fontfamily{#3}\fontseries{#4}\fontshape{#5}%
  \selectfont}%
\fi\endgroup%
\begin{picture}(8070,2776)(1149,-2828)
\put(8776,-2011){\makebox(0,0)[lb]{\smash{{\SetFigFont{12}{14.4}{\rmdefault}{\mddefault}{\updefault}{\color[rgb]{0,0,0}$u_k$}%
}}}}
\put(2524,-1481){\makebox(0,0)[lb]{\smash{{\SetFigFont{12}{14.4}{\rmdefault}{\mddefault}{\updefault}{\color[rgb]{0,0,0}$u$}%
}}}}
\put(1658,-1321){\makebox(0,0)[lb]{\smash{{\SetFigFont{12}{14.4}{\rmdefault}{\mddefault}{\updefault}{\color[rgb]{0,0,0}$x_0$}%
}}}}
\put(1965,-235){\makebox(0,0)[lb]{\smash{{\SetFigFont{12}{14.4}{\rmdefault}{\mddefault}{\updefault}{\color[rgb]{0,0,0}$x_1$}%
}}}}
\put(2931,-302){\makebox(0,0)[lb]{\smash{{\SetFigFont{12}{14.4}{\rmdefault}{\mddefault}{\updefault}{\color[rgb]{0,0,0}$x_2$}%
}}}}
\put(3465,-1455){\makebox(0,0)[lb]{\smash{{\SetFigFont{12}{14.4}{\rmdefault}{\mddefault}{\updefault}{\color[rgb]{0,0,0}$x_3$}%
}}}}
\put(3465,-1808){\makebox(0,0)[lb]{\smash{{\SetFigFont{12}{14.4}{\rmdefault}{\mddefault}{\updefault}{\color[rgb]{0,0,0}$x_4$}%
}}}}
\put(1218,-1715){\makebox(0,0)[lb]{\smash{{\SetFigFont{12}{14.4}{\rmdefault}{\mddefault}{\updefault}{\color[rgb]{0,0,0}$x_k$}%
}}}}
\put(4576,-361){\makebox(0,0)[lb]{\smash{{\SetFigFont{12}{14.4}{\rmdefault}{\mddefault}{\updefault}{\color[rgb]{0,0,0}$x_0$}%
}}}}
\put(5476,-361){\makebox(0,0)[lb]{\smash{{\SetFigFont{12}{14.4}{\rmdefault}{\mddefault}{\updefault}{\color[rgb]{0,0,0}$x_1$}%
}}}}
\put(6076,-361){\makebox(0,0)[lb]{\smash{{\SetFigFont{12}{14.4}{\rmdefault}{\mddefault}{\updefault}{\color[rgb]{0,0,0}$x_2$}%
}}}}
\put(6976,-361){\makebox(0,0)[lb]{\smash{{\SetFigFont{12}{14.4}{\rmdefault}{\mddefault}{\updefault}{\color[rgb]{0,0,0}$x_3$}%
}}}}
\put(7276,-361){\makebox(0,0)[lb]{\smash{{\SetFigFont{12}{14.4}{\rmdefault}{\mddefault}{\updefault}{\color[rgb]{0,0,0}$x_4$}%
}}}}
\put(9076,-361){\makebox(0,0)[lb]{\smash{{\SetFigFont{12}{14.4}{\rmdefault}{\mddefault}{\updefault}{\color[rgb]{0,0,0}$x_k$}%
}}}}
\put(5026,-2011){\makebox(0,0)[lb]{\smash{{\SetFigFont{12}{14.4}{\rmdefault}{\mddefault}{\updefault}{\color[rgb]{0,0,0}$u_1$}%
}}}}
\end{picture}%